\documentclass[11pt]{article}

\usepackage[T1]{fontenc}
\usepackage[utf8]{inputenc}
\usepackage{newpxtext}      
\usepackage{iftex}          
\ifPDFTeX                   
  \usepackage[activate={true,nocompatibility},final]{microtype}
\else                       
  \usepackage[protrusion=true,final]{microtype}
\fi

\usepackage[margin=1in]{geometry}

\usepackage{authblk}                    
\usepackage{etoolbox}
\forcsvlist{\listadd\PaperStatementEnvs}{theorem,lemma,proposition,corollary,definition,remark}
\newcommand{\PaperKeepStatementHead}[1]{%
  \BeforeBeginEnvironment{#1}{\Needspace{4\baselineskip}}}
\forlistloop{\PaperKeepStatementHead}{\PaperStatementEnvs}
\makeatletter
\patchcmd{\@maketitle}{\vskip 2em}{\vskip 0pt}{}%
  {\PackageError{paper}{Could not adjust title top spacing}{}}
\patchcmd{\@maketitle}{\vskip 1.5em}{\vskip 1em}{}%
  {\PackageError{paper}{Could not adjust title-author spacing}{}}
\patchcmd{\@maketitle}{\vskip 1.5em}{\vskip 0pt}{}%
  {\PackageError{paper}{Could not adjust title-abstract spacing}{}}
\patchcmd{\@maketitle}{{\large \@date}}{{\normalsize \@date}}{}%
  {\PackageError{paper}{Could not adjust the date size}{}}
\makeatother

\usepackage{amsmath}
\usepackage{amsthm}         
\usepackage{amssymb,amsfonts}
\usepackage{mathtools}      
\usepackage{newpxmath}      
\usepackage{bm}             
\usepackage{nicefrac}
\allowdisplaybreaks         

\usepackage{thm-restate}    
\usepackage{needspace}      

\usepackage{graphicx}
\graphicspath{{figures/}{./}}                       
\usepackage{booktabs}                               
\usepackage{multirow,makecell,array}
\usepackage[font=small,labelfont=bf]{caption}       
\usepackage{subcaption}                             
\usepackage{algorithm}
\usepackage{algpseudocode}                          

\usepackage{enumitem}
\setlist[itemize]{leftmargin=2.2em,itemsep=2pt,topsep=2pt}
\setlist[enumerate]{leftmargin=2.2em,itemsep=2pt,topsep=2pt}

\usepackage{xcolor}
\definecolor{LinkColor}{rgb}{0.10,0.40,0.75}        
\definecolor{CiteColor}{rgb}{0.70,0.25,0.20}        
\definecolor{UrlColor} {rgb}{0.20,0.50,0.50}        
\definecolor{TodoColor}{rgb}{0.80,0.30,0.10}        

\usepackage{tikz}
\usetikzlibrary{positioning,calc,arrows.meta}

\usepackage[round,sort&compress]{natbib}

\usepackage{url}
\usepackage{hyperref}
\hypersetup{
  colorlinks=true,
  linkcolor=LinkColor,
  citecolor=CiteColor,
  urlcolor=UrlColor,
  breaklinks=true,
  bookmarksnumbered=true,
  hypertexnames=false,       
}
\usepackage{bookmark}                               
\usepackage[capitalise,nameinlink,noabbrev,sort&compress]{cleveref}  

\numberwithin{equation}{section}

\newif\ifdraft \draftfalse
\ifdraft
  \usepackage{lineno}\linenumbers
  \newcommand{\todo}[1]{\textcolor{TodoColor}{\textbf{[TODO:}~#1\textbf{]}}}
\else
  \newcommand{\todo}[1]{}                            
\fi

\newcommand{\N}{\mathbb{N}}

\DeclarePairedDelimiterX{\inner}[2]{\langle}{\rangle}{#1,#2}   

\newcommand{\X}{\mathcal X}
\newcommand{\cH}{\mathcal H}
\newcommand{\cK}{\mathcal K}
\newcommand{\fin}[1]{[#1]^{<\omega}}
\newcommand{\infinite}[1]{[#1]^\omega}
\newcommand{\ew}{\varepsilon}
\DeclareMathOperator{\cont}{c}
\DeclareMathOperator{\wcost}{cost}
\DeclareMathOperator{\code}{code}
\newcommand{\activefamily}{\mathcal A_T}
\newcommand{\closure}[2]{\operatorname{cl}_{#1}(#2)}

\newcommand{\spwidth}{\mathfrak{s}}
\newcommand{\Core}{\operatorname{core}}

\theoremstyle{plain}
\newtheorem{theorem}{Theorem}[section]
\newtheorem{proposition}[theorem]{Proposition}
\newtheorem{lemma}[theorem]{Lemma}
\newtheorem{corollary}[theorem]{Corollary}

\theoremstyle{definition}
\newtheorem{definition}[theorem]{Definition}

\theoremstyle{remark}
\newtheorem{remark}[theorem]{Remark}

\crefname{theorem}{Theorem}{Theorems}          \Crefname{theorem}{Theorem}{Theorems}
\crefname{proposition}{Proposition}{Propositions}
\Crefname{proposition}{Proposition}{Propositions}
\crefname{lemma}{Lemma}{Lemmas}                \Crefname{lemma}{Lemma}{Lemmas}
\crefname{corollary}{Corollary}{Corollaries}   \Crefname{corollary}{Corollary}{Corollaries}
\crefname{conjecture}{Conjecture}{Conjectures} \Crefname{conjecture}{Conjecture}{Conjectures}
\crefname{fact}{Fact}{Facts}                   \Crefname{fact}{Fact}{Facts}
\crefname{definition}{Definition}{Definitions} \Crefname{definition}{Definition}{Definitions}
\crefname{assumption}{Assumption}{Assumptions} \Crefname{assumption}{Assumption}{Assumptions}
\crefname{example}{Example}{Examples}          \Crefname{example}{Example}{Examples}
\crefname{problem}{Problem}{Problems}          \Crefname{problem}{Problem}{Problems}
\crefname{remark}{Remark}{Remarks}             \Crefname{remark}{Remark}{Remarks}
\crefname{claim}{Claim}{Claims}                \Crefname{claim}{Claim}{Claims}
\crefname{algorithm}{Algorithm}{Algorithms}    \Crefname{algorithm}{Algorithm}{Algorithms}

\title{Characterizing Language Generation in the Limit:\\
Finite Witnesses and a Separation-Width Hierarchy}

\author[1]{Xiaoyu Li}
\author[2]{Andi Han}
\author[1]{Jiaojiao Jiang}
\author[2]{Junbin Gao}
\affil[1]{\makebox[14em][l]{University of New South Wales}\quad
\makebox[22em][l]{\texttt{\{xiaoyu.li2,jiaojiao.jiang\}@unsw.edu.au}}}
\affil[2]{\makebox[14em][l]{University of Sydney}\quad
\makebox[22em][l]{\texttt{\{andi.han,junbin.gao\}@sydney.edu.au}}}
\date{September 10, 2026}
\hypersetup{
  pdftitle={Characterizing Language Generation in the Limit: Finite Witnesses and a Separation-Width Hierarchy},
  pdfauthor={Xiaoyu Li, Andi Han, Jiaojiao Jiang, Junbin Gao},
  pdfsubject={Generation in the limit, positive separation, the complete width hierarchy, and dimension obstructions}
}

\begin{document}
\maketitle

\begin{abstract}
Language generation in the limit asks for valid unseen elements from every
exhaustive positive presentation of an unknown infinite language. We
characterize this task for arbitrary families over a countable universe.
Generation is possible exactly when each target can be assigned a finite
positive witness so that the targets activated by any finite sample have
an infinite common intersection. The necessary direction follows from a
universal normalization: a search through unconfirmed histories converts any
successful generator into one depending only on the observed set.
We then ask how large compatible witnesses must be. Positive separation
width records the smallest uniform size bound, with two further levels for
unbounded finite witnesses and the absence of any compatible finite-witness assignment.
Every level occurs. Countable families admit singleton witnesses, explicit
families realize every finite width, and a union of two families with
infinite common cores requires unbounded finite witnesses. Finally,
countable-support and finite-profile obstructions explain why local
combinatorial data cannot determine generation in the limit. The characterization
and full width hierarchy are checked in Lean, including the simplified
normalization and a direct diagonal capture lemma.
The accompanying Lean development is maintained at
\url{https://github.com/xiaoyulics/language-generation-characterization}.

\end{abstract}

\begingroup
\setlength{\intextsep}{4pt}
\begin{figure}[H]
\centering
\begingroup
\definecolor{HeroInk}{HTML}{19344B}
\definecolor{HeroBlue}{HTML}{2C617D}
\definecolor{HeroTeal}{HTML}{187A7B}
\definecolor{HeroRust}{HTML}{A64D32}
\begin{tikzpicture}[
  x=1cm,y=1cm,
  font=\fontsize{9}{10.5}\selectfont,
  text=HeroInk,
  line cap=round,line join=round,
  >={Stealth[length=3.5pt,width=4pt]},
  head/.style={font=\fontsize{10.5}{12}\selectfont\bfseries},
  flow/.style={draw=HeroBlue,line width=.75pt,->},
  token/.style={circle,draw=HeroBlue!65,fill=HeroBlue!5,
    minimum size=5.2mm,inner sep=0pt},
  level/.style={fill=HeroTeal!5,inner sep=2pt,
    font=\fontsize{12}{13}\selectfont,text=HeroTeal}
]
\path[use as bounding box] (0,0) rectangle (16.4,6.2);

\node[head] at (3.35,5.95) {Generation in the limit};
\foreach \x/\letter in {.45/a,1.08/b,1.71/a,2.34/c}
  \node[token] at (\x,5.02) {$\letter$};
\node at (2.94,5.02) {$\cdots$};
\draw[flow] (3.27,5.02) -- (3.72,5.02);
\node[draw=HeroBlue,line width=.8pt,rounded corners=2pt,
  fill=HeroBlue!4,minimum width=12mm,minimum height=8mm,
  align=center,font=\fontsize{9}{10}\selectfont] (gen) at (4.39,5.02)
  {Generator\\$G$};
\draw[flow] (5.04,5.02) -- (5.60,5.02);
\node[text=HeroTeal,font=\fontsize{19}{20}\selectfont] at (5.95,5.03) {$\star$};
\node at (3.30,4.31) {Eventually $G(x_{1:t})\in L\setminus S_t$};
\node[font=\fontsize{8.5}{10}\selectfont] at (3.30,3.91)
  {Every target, every exhaustive positive text};

\node[font=\fontsize{25}{26}\selectfont,text=HeroBlue]
  at (7.35,5.02) {$\Longleftrightarrow$};

\node[head] at (12.35,5.95) {Compatible finite witnesses};
\begin{scope}[shift={(11.72,5.00)}]
  \begin{scope}
    \clip[rotate=14] (0,0) ellipse[x radius=1.77cm,y radius=.56cm];
    \clip[rotate=-14] (0,0) ellipse[x radius=1.77cm,y radius=.56cm];
    \clip (0,0) ellipse[x radius=1.86cm,y radius=.66cm];
    \fill[HeroTeal!7] (-2,-1) rectangle (2,1);
  \end{scope}
  \draw[HeroBlue!85,line width=.75pt,rotate=14]
    (0,0) ellipse[x radius=1.77cm,y radius=.56cm];
  \draw[HeroTeal,line width=.75pt,rotate=-14]
    (0,0) ellipse[x radius=1.77cm,y radius=.56cm];
  \draw[HeroBlue!55,line width=.75pt]
    (0,0) ellipse[x radius=1.86cm,y radius=.66cm];
  \node[font=\fontsize{8}{9}\selectfont] at (-2.04,.30) {$L_1$};
  \node[font=\fontsize{8}{9}\selectfont] at (1.92,.40) {$L_2$};
  \node[font=\fontsize{8}{9}\selectfont] at (1.98,-.43) {$L_3$};
  \draw[HeroInk!55,rounded corners=3pt,line width=.6pt,fill=white]
    (-.93,-.16) rectangle (-.10,.16);
  \foreach \x in {-.72,-.51,-.30}
    \fill[HeroInk] (\x,0) circle[radius=1.35pt];
  \draw[HeroTeal,line width=.7pt] (-.72,0) circle[radius=2.8pt];
  \node[font=\fontsize{8.5}{9}\selectfont] at (-.50,-.34) {$S$};
  \node[text=HeroTeal,font=\fontsize{19}{20}\selectfont] at (.28,.01) {$\star$};
  \node[text=HeroTeal,font=\fontsize{11}{12}\selectfont] at (.84,0) {$\cdots$};
\end{scope}
\node[align=left,text width=21mm,font=\fontsize{8.5}{10}\selectfont]
  at (15.09,5.00) {Infinite\\shared\\intersection};
\draw[HeroInk!45,line width=.5pt,<-] (12.95,5.00) -- (13.96,5.00);
\node at (12.28,3.91) {$T(L)\subseteq S\subseteq L$ activates $L$};

\draw[HeroInk!18,line width=.5pt] (.1,3.47) -- (16.3,3.47);
\node[head,anchor=west] at (.10,3.08) {Positive separation width $\spwidth(\cH)$};
\node[font=\fontsize{8.5}{10}\selectfont,anchor=east] at (16.30,3.08)
  {Every level occurs};
\fill[HeroTeal!5,rounded corners=3pt] (.10,.08) rectangle (11.98,2.62);
\fill[HeroRust!5,rounded corners=3pt] (12.27,.08) rectangle (16.30,2.62);
\node[text=HeroTeal,font=\fontsize{9}{10}\selectfont\bfseries]
  at (6.02,2.24) {Generation possible};
\node[text=HeroRust,font=\fontsize{9}{10}\selectfont\bfseries]
  at (14.29,2.24) {Impossible};
\draw[HeroTeal!65,line width=.8pt] (.75,1.53) -- (10.20,1.53);
\foreach \x/\value in {.85/0,2.35/1,3.85/2,5.35/3,7.05/\cdots,10.12/\omega}
  \node[level] at (\x,1.53) {$\value$};
\draw[HeroInk!22,densely dotted,line width=.6pt] (8.13,.30) -- (8.13,1.92);
\node[level,text=HeroRust,fill=HeroRust!5] at (14.29,1.53) {$\omega+1$};
\node[align=center] at (4.05,.69) {Uniformly bounded\\finite witnesses};
\node[align=center] at (10.06,.69) {Finite witnesses;\\no uniform bound};
\node[align=center] at (14.29,.69) {No compatible\\finite witnesses};
\end{tikzpicture}
\endgroup
\caption{\textbf{From generation to finite witnesses, then to witness size.}
Generation is possible exactly when one finite-witness assignment makes
each nonempty family of active targets have infinite full intersection,
at every finite sample $S$. Activation means $T(L)\subseteq S\subseteq L$.
In the upper sketch, dots form the observed set $S$ ($S_t$ at time $t$),
rings mark witness points, and stars mark unseen outputs. Below, every
width occurs; $\omega$ permits finite witnesses with no uniform size bound.}
\label{fig:overview}
\end{figure}
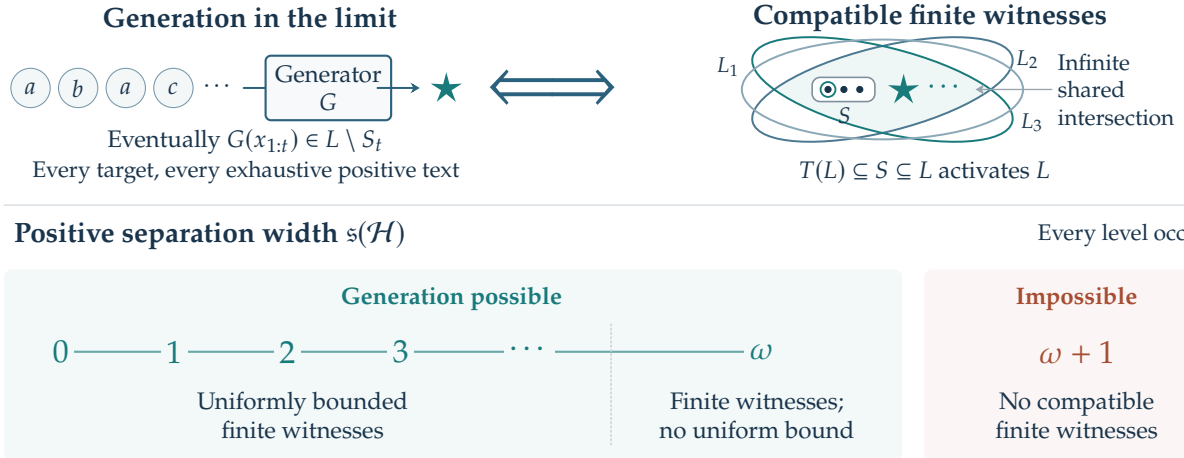
\endgroup

\clearpage
\pdfbookmark[0]{Contents}{contents}
\begingroup
\small
\makeatletter
\patchcmd{\l@section}{1.0em}{0.55em}{}%
  {\PackageError{paper}{Could not compact contents spacing}{}}
\makeatother
\tableofcontents
\endgroup
\clearpage

\section{Introduction}
\label{sec:intro}

A generator observes positive examples from an unknown language and must
produce another valid example. In generation in the limit, the observations
arrive in an arbitrary order and eventually enumerate the entire infinite
target. Success means that, after some finite time, every output belongs to
the target and has not appeared among the inputs. There is no feedback
about whether an output is correct.

\citet{kleinberg2024generation} showed that every countable family of infinite
languages admits such a generator, even when identifying the target is
impossible. For arbitrary families over a countable universe,
\citet{raman2025generation} characterized uniform generation by finite closure
dimension and non-uniform generation by increasing countable covers of finite
closure dimension. For ordinary generation, whose convergence time may depend
on the entire presentation, they gave sufficient conditions and asked for a
complete characterization.

Our first result answers this question by identifying the finite positive
evidence that makes generation possible. Our second result studies the size
of that evidence and shows why pointwise finiteness cannot be replaced by
any uniform finite bound.

\paragraph{Finite witnesses: when generation is possible.}
Assign each target $L$ a finite set $T(L)\subseteq L$. At a finite sample
$S$, call $L$ \emph{active} if $T(L)\subseteq S\subseteq L$: the sample is
consistent with $L$ and has already revealed its witness.
\Cref{thm:characterization} states that generation in the limit is possible
exactly when one assignment makes the full intersection of every nonempty
active family infinite.

The sufficient direction explains the condition. Every target eventually
becomes active on each of its texts. From then on, any fresh point common
to all active targets is a valid output. Several targets may remain
indistinguishable indefinitely; generation needs only a common source of
new elements.

The necessary direction must handle a generator's dependence on order and
repetitions. \Cref{thm:normalization} gives a universal transformation from
any sequence-input generator $G$ to a set-input generator $g$, preserving
every infinite target on which $G$ succeeds. The new generator searches
histories whose outputs have not yet been observed. Finitely many positive
confirmations force this search to follow a canonical finite chain of genuine
errors and then pass beyond it. This produces witnesses satisfying
\[
 T(L)\subseteq S\subseteq L,\qquad S\text{ finite}
 \quad\Longrightarrow\quad g(S)\in L\setminus S.
\]
The construction uses no knowledge of the target and no correctness oracle.
With an effective universe, this locking conclusion also characterizes
computable generation of a single target by the existence of an infinite
computably enumerable subset; see Appendix~\ref{app:computability}.

\paragraph{Witness size: a complete separation-width hierarchy.}
Call a nonempty subfamily \emph{bad} when its full common intersection is
finite. The witness condition is equivalent to requiring that each bad
subfamily contain $L,K$ with $T(L)\not\subseteq K$. We call this
\emph{positive separation}.

Positive separation width records the least uniform finite bound on the
sizes of separating witnesses, when one exists. Otherwise it is $\omega$
if a pointwise finite assignment exists, and $\omega+1$ if none exists.
Here $\omega$ is the first infinite ordinal and $\omega+1$ its successor;
\cref{sec:separation} explains the scale. The core characterization has the
immediate numerical form
\[
 \cH\text{ generatable in the limit}
 \quad\Longleftrightarrow\quad \spwidth(\cH)\le\omega.
\]

The additional content is the hierarchy in \cref{thm:hierarchy}.
Every countable family admits distinct singleton witnesses. Nevertheless,
explicit families have arbitrarily large finite width, with matching upper
and lower bounds. A family of languages containing either of two fixed
disjoint infinite blocks has width $\omega$. Thus even a union of two
classes with infinite common cores can require unbounded finite witnesses.
These are witness-size statements; an arbitrary text may delay a designated
witness for as long as it wishes.

The finite-level lower bounds must allow witnesses to depend on entire
targets, including their infinite tails. A direct diagonal capture lemma
handles this dependence before a finite incidence count finishes the proof.
The same lemma yields forced witness growth along an explicit chain of
targets in the width-$\omega$ example.

\paragraph{Why local dimensions miss the distinction.}
The compatibility of witnesses is global. Every countable family is
generatable, but the family of all infinite subsets is not. Consequently,
a dimension whose infinity is always supported by a countable subclass
cannot characterize generation in the limit by finiteness. Even the complete
finite trace and positive-closure profiles fail to distinguish generatable
cofinite targets from the class of all infinite targets.
\Cref{sec:barriers} makes these limitations precise without excluding
invariants that retain additional global structure.

\subsection{Relation to earlier work}

The positive-data setting originates in language identification in the limit
\citep{gold1967identification}. Finite tell-tales are central to
\citet{angluin1980inference}'s characterization for indexed recursive-language
families. Our witnesses have a different compatibility requirement: several
active targets may remain indistinguishable if their intersection supplies
infinitely many outputs. A witness need not identify a target or exclude all
its proper sublanguages.
For a broader account of language generation and its variants, see
\citet[Part~II]{mehrotra2026thesis}.

Set-drivenness and locking normal forms have also been studied in language
identification \citep{kotzing2017normal}. Our normalization concerns the
single-element generation criterion and uses the test of whether a simulated
output has appeared among the observations. Sorting the inputs before
querying an arbitrary successful generator does not suffice;
\cref{app:sorting} gives a counterexample.

The two-core class is a union of two classes with infinite fixed common
subsets, so its generatability also follows from the finite-union sufficient
condition of \citet[Corollary~14]{raman2025generation}. Its role here is the necessity of
unbounded finite witnesses. Our characterization also recovers the
countable-family and increasing-EUC-cover sufficient conditions by direct
assignments in \cref{sec:consequences}. These recoveries establish existence;
the singleton bound and the hierarchy then quantify the required evidence.
Arbitrary unions require care:
\citet{hanneke2025union} and \citet{bai2026noise} give a uniformly generatable
class and a non-uniformly generatable class whose union is not generatable.
Assignments that work separately need not stay valid when targets from both
classes become active at one sample.

Other characterizations place requirements on the range of generated outputs.
\citet{kalavasis2026breadth} characterize generation with several notions of
breadth, while \citet{kleinberg2025density} quantify coverage using density.
These questions concern how much of a target is generated. Separation width
instead measures the size of positive evidence compatible across targets;
the infinite-intersection condition guarantees a supply of fresh elements
without prescribing their density.

The observation model also matters. \citet{kleinberg2026partial} study partial
enumeration, where only an infinite subset of the target is revealed.
\citet{li2026contrastive} study unordered pairs with opposite labels and
characterize uniform contrastive generation by a contrastive closure
dimension. In the nested model of \citet{li2026flood}, a verifier recognizes
an ambient formal language while the examples reveal part of an unknown
valuable sublanguage; validity and valuable coverage then become separate
requirements. Our finite witnesses use exhaustive positive observations of
individual target elements. Applying the same idea to these other observation
models would require a corresponding notion of witness compatibility.

We define the model and state the characterization in \cref{sec:prelim},
prove normalization and the characterization in
\cref{sec:normalization,sec:characterization-proof}, and recover known
sufficient conditions in \cref{sec:consequences}. We then develop separation
width and its hierarchy in \cref{sec:separation,sec:hierarchy}.
\Cref{sec:barriers,sec:discussion} explain the structural and computational
boundaries. The appendices give a complete normalization example, stronger
witness divergence, the padding obstruction, and the formalization account.
Appendices~\ref{app:quantifiers} and~\ref{app:computability} explain the
quantifier distinctions and develop the computable-generation consequences.

\section{The model and the characterization}
\label{sec:prelim}

Let $\X$ be a countably infinite set. Write $\fin{\X}$ for its finite subsets,
including the empty set, and $\infinite{\X}$ for its infinite subsets. A
\emph{language family} is any $\cH\subseteq\infinite{\X}$. No effectiveness,
cardinality, or descriptive-set-theoretic assumption is imposed on $\cH$.
The set of finite words over $\X$ is $\X^{<\omega}$, with empty word $\ew$.
For a word $\sigma=(\sigma_1,\ldots,\sigma_m)$, its \emph{content} is the set
$\cont(\sigma)=\{\sigma_i:1\le i\le m\}$ of elements that occur in the word,
with $\cont(\ew)=\varnothing$.
Content ignores order and repetition; for example,
$\cont((a,b,a))=\{a,b\}$.
We write $\sigma\prec\tau$ for strict prefix and $|\sigma|$ for word length.
Our conventions are $\N=\{1,2,\ldots\}$ and $\N_0=\N\cup\{0\}$.

A \emph{text} for $L$ is an infinite sequence $x=(x_t)_{t\ge1}$ with
$\{x_t:t\ge1\}=L$. Repetitions and arbitrarily long delays are permitted.
Let $x_{1:t}=(x_1,\ldots,x_t)$ and $S_t=\cont(x_{1:t})$.

\begin{definition}[Generation in the limit]\label{def:ordinary}
A total function $G:\X^{<\omega}\to\X$ \emph{generates $L$ in the limit} if
for every text $x$ for $L$ there is a time $t_0=t_0(L,x)$ such that
\begin{equation}\label{eq:ordinary}
  G(x_{1:t})\in L\setminus S_t
  \qquad\text{for every }t\ge t_0.
\end{equation}
The family $\cH$ is \emph{generatable in the limit} if one total $G$ has this
property for every $L\in\cH$. A generator is \emph{set-driven} if it has the
form $G(\sigma)=g(\cont(\sigma))$ for some $g:\fin{\X}\to\X$.
\end{definition}

The output is a single element. Freshness refers to the observed inputs,
not to earlier outputs. The generator has no target-membership oracle and
receives no correctness feedback. We work with arbitrary total functions;
computability is discussed separately in \cref{cor:computable} and
Appendix~\ref{app:computability}.
Requiring freshness on every finite input, rather than only eventually on
each text, gives the same existence notion by the repair in
\cref{eq:fresh-repair}.

\begin{definition}[Finite-witness condition]\label{def:witness}
A \emph{witness assignment} for $\cH$ is a function
$T:\cH\to\fin{\X}$ with $T(L)\subseteq L$ for every $L\in\cH$.
At a finite set $S\subseteq\X$, define
\begin{equation}\label{eq:active}
  \activefamily(S)=\{L\in\cH:T(L)\subseteq S\subseteq L\}.
\end{equation}
The assignment satisfies the \emph{finite-witness condition} if
\begin{equation}\label{eq:witness}
  \activefamily(S)\ne\varnothing
  \quad\Longrightarrow\quad
  \left|\bigcap_{L\in\activefamily(S)}L\right|=\infty
  \qquad\text{for every }S\in\fin{\X}.
\end{equation}
\end{definition}

\begin{theorem}[Finite-witness characterization]\label{thm:characterization}
For every countably infinite $\X$ and every
$\cH\subseteq\infinite{\X}$, the following are equivalent:
\begin{enumerate}[label=\textup{(\roman*)}]
  \item $\cH$ is generatable in the limit.
  \item $\cH$ is generatable in the limit by a set-driven generator.
  \item $\cH$ admits a witness assignment satisfying \eqref{eq:witness}.
\end{enumerate}
\end{theorem}

The intersection in \eqref{eq:witness} is over all active languages
simultaneously. Even requiring every finite active subfamily to have
infinite intersection does not suffice: Appendix~\ref{app:quantifiers} gives a
nongeneratable class satisfying that weaker condition with empty witnesses.
The empty active family imposes no requirement, so no convention for its
intersection is needed.

A finite witness concerns \emph{which} examples have appeared. Even an
arbitrarily large finite subset of a target may omit a specified witness
element. Thus the theorem does not replace convergence in the limit by a
target-dependent bound on the number of distinct examples. Nor does the
existence of an assignment supply a procedure to compute it or to test the
intersection condition.

\section{Set-driven normalization}
\label{sec:normalization}

The obstacle is that a generator may use the order and repetitions of its
input. We remove this dependence by simulating histories over the observed
set. An output outside that set is \emph{unconfirmed}: it may be an error,
or a valid point that has not yet appeared. Positive observations eventually
eliminate the latter possibility for any fixed finite collection of histories.

\begin{theorem}[Universal normalization]\label{thm:normalization}
For every total $G:\X^{<\omega}\to\X$, there is a total
$g:\fin{\X}\to\X$ such that every infinite language $L$ generated by $G$
has a finite $T_L\subseteq L$ satisfying
\begin{equation}\label{eq:locking}
 T_L\subseteq S\subseteq L,\qquad S\text{ finite}
 \quad\Longrightarrow\quad g(S)\in L\setminus S.
\end{equation}
The same $g$ works for all such $L$ and uses only finitely many evaluations
of $G$ on each input $S$.
\end{theorem}

\subsection{The construction}
\label{sec:construction}

Fix a repetition-free enumeration of $\X$; all minima below refer to this
order. Let $E_k(A)$ be the first $k$ points of $A$ in this order, or all
of $A$ if it has fewer than $k$ points. Thus $|E_k(L)|=k$ for infinite $L$.

Fix also an injective numbering $\code:\X^{<\omega}\to\N_0$, assigning
distinct natural numbers to distinct finite words. Such a numbering exists
because $\X^{<\omega}$ is countable: for example, list its words without
repetition as $w_0,w_1,\ldots$ and set $\code(w_n)=n$.
We call $\code(\sigma)$ the \emph{code} of $\sigma$; ``least-code'' means
having the smallest code among the candidates. Each word has only finitely
many predecessors in this order, since only finitely many natural numbers
are smaller than its code.

Repair outputs already present in the input by setting
\begin{equation}\label{eq:fresh-repair}
 F(\sigma)=
 \begin{cases}
 G(\sigma),&G(\sigma)\notin\cont(\sigma),\\
 \min(\X\setminus\cont(\sigma)),&G(\sigma)\in\cont(\sigma).
 \end{cases}
\end{equation}
Thus $F$ is fresh on every history and eventually agrees with $G$ on each
text where $G$ succeeds. Any fixed fresh fallback can replace the minimum.

Given $S$, start at the empty history. At round $k$, take the least-code
strict extension over $S$ that contains $E_k(S)$ and whose output remains
unconfirmed by $S$. Perform exactly $|S|$ rounds.

\begin{algorithm}[H]
\caption{Set-driven normalization by unconfirmed extensions}
\label{alg:normalization}
\begin{algorithmic}[1]
\Require $S\in\fin{\X}$; $F$ from \eqref{eq:fresh-repair}; fixed enumerations
\State $\sigma\gets\ew$
\For{$k=1,\ldots,|S|$}
 \State Choose the least-code $\tau\in S^{<\omega}$ satisfying
 \Statex \hspace{\algorithmicindent}\hspace{\algorithmicindent}
       $\sigma\prec\tau$, $E_k(S)\subseteq\cont(\tau)$, and $F(\tau)\notin S$
 \State $\sigma\gets\tau$
\EndFor
\State \Return $F(\sigma)$
\end{algorithmic}
\end{algorithm}

Every round has a candidate: append a listing of all of $S$ to the current
history. Its content is exactly $S$, so its $F$-output lies outside $S$.
The listing is nonempty whenever a round is performed. It suffices to
inspect words up to the code of this appended history, so the search
uses only finitely many evaluations. The output is fresh
against the whole sample; for $S=\varnothing$, the loop is empty and returns
$F(\ew)$.

\subsection{Why finite confirmations suffice}
\label{sec:normalization-proof}

\begin{proof}[Proof of \cref{thm:normalization}]
Fix an infinite target $L$ generated by $G$, hence also by $F$.
For the proof only, construct the \emph{genuine-error chain}:
$\sigma_0=\ew$, and at stage $j\ge1$ choose the least-code word satisfying
\begin{equation}\label{eq:genuine-errors}
 \sigma_{j-1}\prec\sigma_j,\qquad
 \cont(\sigma_j)\subseteq L,\qquad
 E_j(L)\subseteq\cont(\sigma_j),\qquad F(\sigma_j)\notin L.
\end{equation}
Stop when no such word exists. An infinite chain would give an exhaustive
text for $L$ with errors at the strictly increasing times $|\sigma_j|$,
contradicting success. Let $m$ be its length. The stopping condition says
\begin{equation}\label{eq:terminal-safety}
 \sigma_m\prec\tau,\quad \cont(\tau)\subseteq L,\quad
 E_{m+1}(L)\subseteq\cont(\tau)
 \quad\Longrightarrow\quad F(\tau)\in L.
\end{equation}

Choose the finite positive witness
\begin{equation}\label{eq:normalization-witness}
 T_L=E_{m+1}(L)\cup\cont(\sigma_m)\cup
 \bigcup_{j=1}^{m}\{F(\tau):\code(\tau)<\code(\sigma_j),\ F(\tau)\in L\}.
\end{equation}
Each code has finitely many predecessors, so this union is finite.
Its first component already ensures $|T_L|\ge m+1$.

Fix any finite $T_L\subseteq S\subseteq L$. Then
$E_j(S)=E_j(L)$ for $j\le m+1$. We claim that the first $m$ rounds of
\cref{alg:normalization} select exactly $\sigma_1,\ldots,\sigma_m$.
Inductively, $\sigma_j$ is an eligible extension of $\sigma_{j-1}$:
its content lies in $\cont(\sigma_m)\subseteq S$, its checkpoint agrees,
and its output lies outside $L$, hence outside $S$.
Any earlier-code candidate $\tau$ with $F(\tau)\in L$ has its output
included in $T_L$ by \eqref{eq:normalization-witness}, contradicting
$F(\tau)\notin S$. Any earlier candidate with $F(\tau)\notin L$ would
contradict the least-code choice in \eqref{eq:genuine-errors}.
This proves the claim for every such $S$.

Since $|S|\ge m+1$, the algorithm performs another round. Its final history
$\rho$ therefore strictly extends $\sigma_m$, lies over $L$, and contains
$E_{m+1}(L)$. Equation~\eqref{eq:terminal-safety} gives $F(\rho)\in L$,
while the final selection gives $F(\rho)\notin S$.
This proves \eqref{eq:locking}, including $m=0$.
The algorithm itself used only $G$ and $S$; the target and its genuine
errors entered only the proof.
\end{proof}

\begin{corollary}[Computability preservation]\label{cor:computable}
With an effective coding of $\X$, a total computable $G$ has a total
computable normalization $g$ as in \cref{thm:normalization}.
\end{corollary}
\begin{proof}
Choose effective point and word enumerations. Each candidate test is
decidable from $S$ and a call to $G$, and each code search terminates by
the append-all argument. There are only $|S|$ rounds.
\end{proof}

Finite positive confirmations force the simulation to follow the
genuine-error chain and pass beyond it. Later histories need not stabilize;
Appendix~\ref{app:worked-normalization} illustrates the mechanism.

\section{From normalization to the characterization}
\label{sec:characterization-proof}

Normalization supplies witnesses with a uniform guarantee over all finite
extensions inside a target. This directly gives the necessary compatibility
between witnesses of different targets.

\begin{proof}[Proof of \cref{thm:characterization}]
For \textup{(i)}$\Rightarrow$\textup{(iii)}, apply
\cref{thm:normalization} to a generator for $\cH$, obtaining one $g$ and
witnesses $T(L)$ satisfying \eqref{eq:locking}. Fix a finite $S$ with
$\activefamily(S)\ne\varnothing$, and put
\[
 C=\bigcap_{L\in\activefamily(S)}L.
\]
If $C$ were finite, every active $L$ would satisfy
$T(L)\subseteq S\subseteq C\subseteq L$. Applying \eqref{eq:locking}
at the \emph{same} input $C$ gives $g(C)\in L\setminus C$ for all these
$L$, forcing $g(C)\in C\setminus C$. Thus $C$ is infinite.
This argument uses the full active family, regardless of its cardinality.

For \textup{(iii)}$\Rightarrow$\textup{(ii)}, define
\begin{equation}\label{eq:generator-from-witnesses}
 g(S)=
 \begin{cases}
 \displaystyle\min\left(\bigcap_{L\in\activefamily(S)}(L\setminus S)\right),
       &\activefamily(S)\ne\varnothing,\\[5pt]
 \min(\X\setminus S),&\activefamily(S)=\varnothing.
 \end{cases}
\end{equation}
The assumed infinitude makes this a total function. On any text for $L$,
its finite witness $T(L)$ eventually appears. Thereafter $L$ remains active,
so every output lies in $L\setminus S$.
Finally, \textup{(ii)}$\Rightarrow$\textup{(i)} is immediate by setting
$G(\sigma)=g(\cont(\sigma))$.
\end{proof}

The same-input argument explains why the common intersection must be
infinite: a finite common core would itself be a legal input beyond all
active witnesses, with no common fresh output left.

\section{Recovering known sufficient conditions}
\label{sec:consequences}

The characterization can be used without first designing a generator:
assign finite positive witnesses and check their active intersections.
We illustrate this with two direct constructions that recover known
sufficient conditions. The first uses the order of a countable list;
the second uses the order of an increasing cover.

\subsection{Countable families}
\label{sec:countable-recovery}

The following recovers the information-theoretic countable-family theorem
of \citet[Theorem~4.1]{kleinberg2024generation}.

\begin{corollary}[Countable families]\label{cor:countable}
Every countable family of infinite languages satisfies the finite-witness
condition and is generatable in the limit.
\end{corollary}
\begin{proof}
List the distinct languages as $L_1,L_2,\ldots$, using a finite list if
needed; the empty family is vacuous. For each $i$ and each $j<i$ with
$L_i\not\subseteq L_j$, choose a point $w_{ij}\in L_i\setminus L_j$.
Include all these points in $T(L_i)$, and enlarge this finite set within
$L_i$ until $|T(L_i)|\ge i$.

Fix a finite sample $S$ with a nonempty active family. Every active index
satisfies $i\le|T(L_i)|\le|S|$, so there is a largest active index $m$.
If $j<m$ is active and $L_m\not\subseteq L_j$, then
$w_{mj}\in T(L_m)\subseteq S\subseteq L_j$, contradicting its choice.
Thus $L_m\subseteq L_j$ for every active $j$. Since $L_m$ itself is
active, the common intersection is exactly $L_m$, which is infinite.
Apply \cref{thm:characterization}.
\end{proof}

The padding makes the active list finite; the disagreement points make
its last member a common infinite core. No uniform upper bound on witness
size is required here. In \cref{thm:singleton} we will strengthen this conclusion by showing
that every countable family admits distinct singleton witnesses.

For a single target, the empty witness suffices. More generally, if a
family has an infinite common subset, setting $T(L)=\varnothing$ for
every target already satisfies \eqref{eq:witness}.

\subsection{Increasing covers with eventually unbounded closure}
\label{sec:euc-recovery}

For a language family $\cK$ and a finite set $S$ contained in at least one
member of $\cK$, define its positive closure by
\begin{equation}\label{eq:positive-closure}
  \closure{\cK}{S}
  =\bigcap\{K\in\cK:S\subseteq K\}.
\end{equation}
The family inside the intersection is the \emph{version space}: the
targets consistent with $S$. If $F\subseteq S$ and this family is nonempty, then
$\closure{\cK}{F}\subseteq\closure{\cK}{S}$.
We use the following finite-prefix formulation of the eventually unbounded
closure property of \citet[Definition~40]{raman2025generation}.

\begin{definition}[Eventually unbounded closure]\label{def:euc}
A family $\cK$ has \emph{eventually unbounded closure} \textup{(EUC)} if
for every $L\in\cK$ there is a finite $F\subseteq L$ such that
$\closure{\cK}{F}$ is infinite.
\end{definition}

This formulation is equivalent to requiring an infinite closure eventually
along every text of every member. One direction follows by taking a prefix
of any such text. Conversely, every text eventually contains a fixed finite
$F$, after which closure monotonicity preserves infinitude. This equivalence
uses only texts that enumerate a member of $\cK$.

For an EUC family alone, choosing $T(L)=F(L)$ from \cref{def:euc}
already works: at a nonempty active sample $S$, its full version-space
closure is infinite and is contained in the active intersection.
The next construction extends this observation to an increasing cover,
recovering \citet[Theorem~44]{raman2025generation}.

\begin{corollary}[Increasing EUC covers]\label{cor:euc}
Suppose
$\cH=\bigcup_{n\ge1}\cH_n$, where
$\cH_1\subseteq\cH_2\subseteq\cdots$ and every $\cH_n$ has EUC.
Then $\cH$ satisfies the finite-witness condition and is generatable in the limit.
\end{corollary}
\begin{proof}
For each $L\in\cH$, choose $n(L)$ with $L\in\cH_{n(L)}$.
By EUC, choose a finite $F(L)\subseteq L$ for which
$\closure{\cH_{n(L)}}{F(L)}$ is infinite.
Extend $F(L)$ to a finite $T(L)\subseteq L$ with
$|T(L)|\ge n(L)$.

Fix $S$ with a nonempty active family.
For each active $L$ we have $n(L)\le|T(L)|\le|S|$, so the active layer
indices have a maximum $m$. Choose an active $L_*$ with $n(L_*)=m$.
Because $F(L_*)\subseteq S\subseteq L_*$, the version space of $\cH_m$
at $S$ is nonempty, and
\[
  \closure{\cH_m}{F(L_*)}\subseteq\closure{\cH_m}{S}.
\]
The right-hand side is therefore infinite. Increasingness of the cover
puts every active language in $\cH_m$, whence
\[
  \closure{\cH_m}{S}\subseteq
  \bigcap_{L\in\activefamily(S)}L.
\]
The active intersection is infinite. Apply \cref{thm:characterization}.
\end{proof}

The two constructions use the same idea: the witness size bounds the
indices that can be active at a finite sample, and the largest active
index supplies an infinite common core. The EUC construction allows each
layer to contain uncountably many targets.

EUC itself is not necessary. The cofinite subsets of $\X$ form a countable,
generatable family, but their positive closure at every finite sample $S$
is exactly $S$: for each $x\notin S$, the target $\X\setminus\{x\}$
is consistent and omits $x$. The active-family criterion succeeds by
retaining the targets whose assigned witnesses have appeared, even when
the full version-space closure is finite. The increasing-cover condition
above is sufficient; we make no necessity claim for that condition.

\section{Positive separation and its width}
\label{sec:separation}

The preceding constructions establish the existence of compatible finite
witnesses. We now ask how small they can be. A reformulation in terms of
separating languages inside subfamilies with finite common cores makes
that question easier to study.
For nonempty $\mathcal F\subseteq\cH$, write
$\Core(\mathcal F)=\bigcap_{L\in\mathcal F}L$ and call $\mathcal F$
\emph{bad} if this intersection is finite.

\begin{definition}[Positive separation]\label{def:separation}
A positive assignment $P$, with $P(L)\subseteq L$ for every $L\in\cH$,
\emph{separates} $\cH$ if
\begin{equation}\label{eq:separation}
 \forall\text{ nonempty bad }\mathcal F\subseteq\cH,
 \quad \exists L,K\in\mathcal F:\ P(L)\not\subseteq K.
\end{equation}
The values of $P$ are allowed to be infinite unless stated otherwise.
\end{definition}

Thus a bad subfamily must contain a witness point, assigned to one of its
members, that another member omits. One assignment must meet all these
constraints simultaneously.

\begin{proposition}[Separation equivalence]\label{prop:separation}
A positive assignment $T$ with finite values satisfies \eqref{eq:witness} if and only if
it separates $\cH$. For any fixed positive assignment, it is equivalent in
\eqref{eq:separation} to test only countable bad subfamilies.
\end{proposition}
\begin{proof}
If separation fails for a nonempty bad $\mathcal F$, then
$T(L)\subseteq K$ for every $L,K\in\mathcal F$. Consequently
$T(L)\subseteq\Core(\mathcal F)$ for every $L\in\mathcal F$.
At the finite sample $S=\Core(\mathcal F)$ all these languages are active.
Their full active intersection is contained in $S$, violating
\eqref{eq:witness}.

Conversely, a nonempty active family with finite intersection is itself a
bad subfamily. Every one of its witnesses lies in $S$, and every one of its
languages contains $S$. Hence it is not separated.

For the countable reduction, fix a nonempty $\mathcal F$ and one member
$L_*\in\mathcal F$. For each $x\notin\Core(\mathcal F)$ choose
$K_x\in\mathcal F$ omitting $x$. The at-most-countable family
$\{L_*\}\cup\{K_x:x\notin\Core(\mathcal F)\}$ has exactly the same
intersection. If the original family is unseparated by a fixed assignment,
so is this subfamily.
\end{proof}

Countable reduction does not exchange the quantifiers
$\exists T\,\forall\mathcal F$ and $\forall\mathcal F\,\exists T_{\mathcal F}$.
Assignments valid on separate countable subfamilies need not fit together
on the whole class. Appendix~\ref{app:quantifiers} gives a single
counterexample that also explains why finite active intersections cannot
replace the full intersection.

\begin{theorem}[Singleton witnesses for countable families]
\label{thm:singleton}
Every countable family of infinite languages admits a separating singleton
assignment $T(L)=\{p_L\}$ with all points $p_L$ distinct.
\end{theorem}
\begin{proof}
List the distinct targets as $L_1,L_2,\ldots$, using a finite list if needed.
At step $i$, let $B_i$ be the union of those common intersections of
nonempty subfamilies of $\{L_1,\ldots,L_i\}$ that are finite sets:
\[
 B_i=\bigcup_{\substack{\varnothing\ne I\subseteq\{1,\ldots,i\}\\
                         |\bigcap_{j\in I}L_j|<\infty}}
                         \ \bigcap_{j\in I}L_j.
\]
There are only finitely many index sets $I$, so $B_i$ is finite. Choose
\[
 p_i\in L_i\setminus\bigl(B_i\cup\{p_1,\ldots,p_{i-1}\}\bigr).
\]
If a bad subfamily is finite, let $i$ be its largest index. Its core is
contained in $B_i$ and therefore omits $p_i$; some member of the subfamily
omits this assigned point. If a bad subfamily is infinite, its infinitely
many distinct assigned points cannot all lie in its finite core. It too
is separated. The empty class is vacuous.
\end{proof}

This sharpens the construction in \cref{cor:countable} from arbitrary
finite witnesses to singletons. It concerns the size of a certificate,
not the time needed to observe that certificate: a text may postpone its
single designated point arbitrarily long.

To measure witness size, we use the ordered scale
\[
 0<1<2<\cdots<\omega<\omega+1.
\]
Here $\omega$ is the first infinite ordinal, the order type of
$0<1<2<\cdots$; its successor $\omega+1$ appends a greatest element after
that sequence. Both have countably infinite underlying sets. The distinction
concerns their order types. In particular, $\sup_{n\in\N_0}n=\omega$.

\begin{samepage}
\begin{definition}[Positive separation width]\label{def:width}
The \emph{positive separation width} $\spwidth(\cH)$ is defined by three cases:
\begin{enumerate}[label=\textup{(\roman*)}]
 \item If a separating assignment has all its witness sizes bounded by
       one finite integer, the width is the least such bound.
 \item If a pointwise finite separating assignment exists, but no
       separating assignment has a uniform finite bound, the width is $\omega$.
 \item If no pointwise finite separating assignment exists, the width is
       $\omega+1$.
\end{enumerate}
\end{definition}
\end{samepage}

Thus width $\omega$ allows each target a finite witness while requiring
unbounded sizes across targets. Width $\omega+1$ means that every separating
assignment has at least one infinite value. The empty class has width zero.

\begin{corollary}[Generation threshold]\label{thm:width}
\label{cor:width-characterization}
For every $\cH\subseteq\infinite{\X}$,
\begin{equation}\label{eq:width-characterization}
 \cH\text{ is generatable in the limit}
 \quad\Longleftrightarrow\quad \spwidth(\cH)\le\omega.
\end{equation}
\end{corollary}
\begin{proof}
By definition, width at most $\omega$ means that a pointwise finite separating
assignment exists. Apply \cref{prop:separation,thm:characterization}.
\end{proof}

At width $\omega$, each target's finite witness still appears by some finite
time on every text for that target. Unbounded sizes across targets are
therefore compatible with convergence on each text.

\begin{proposition}[Basic properties]\label{prop:width-basic}
The width is invariant under bijective relabeling of $\X$ and monotone
under passing to subfamilies. It is zero exactly when the full common
intersection is infinite, with $\Core(\varnothing)=\X$.
\end{proposition}
\begin{proof}
Restriction preserves separation and witness-size bounds; bijections
preserve inclusions and cardinalities. Width zero means the empty assignment
separates. This holds exactly when there is no nonempty bad subfamily,
equivalently when the full common intersection is infinite.
\end{proof}

\begin{remark}[Equivalent optimization formula]\label{rem:width-cost}
Define the witness cost by $\wcost(W)=|W|$ for finite $W$, and
$\wcost(W)=\omega+1$ for infinite $W$.
Then
\begin{equation}\label{eq:width}
 \spwidth(\cH)=\min_{P\text{ separates }\cH}\ \sup_{L\in\cH}\wcost(P(L)),
\end{equation}
with empty supremum zero. A feasible assignment always exists: $P(L)=L$
separates every bad subfamily. Each assignment has cost in
$\N_0\cup\{\omega,\omega+1\}$, so the minimum is attained, and its three
possible regimes are exactly those in \cref{def:width}.
The charge $\omega+1$ is a cost convention, not the cardinality of an
infinite witness. Charging $\omega$ would merge the cost of an infinite
witness with the supremum of unbounded finite witness sizes.
\end{remark}

\begin{samepage}
\section{The complete separation-width hierarchy}
\label{sec:hierarchy}

The characterization guarantees finite witnesses for individual targets.
Their required sizes have a richer structure: every finite level occurs,
and some generatable classes require unbounded finite witnesses.

\begin{theorem}[Separation-width hierarchy]\label{thm:hierarchy}
\label{cor:full-range}
On every countably infinite universe $\X$:
\begin{enumerate}[label=\textup{(\roman*)}]
 \item Every countable language family has width at most one.
 \item For every $d\in\N_0\cup\{\omega\}$, some generatable family has
       width exactly $d$.
 \item The family $\infinite{\X}$ of all infinite subsets has width $\omega+1$.
\end{enumerate}
Thus every value of the width occurs.
\end{theorem}
\end{samepage}

The finite levels come from an incidence-counting construction. A class
containing these examples for every finite size then realizes $\omega$.
The only auxiliary selection fact needed for the lower bounds is the
following diagonal lemma.

\subsection{A bounded-capture lemma}

The lower bound must handle assignments that depend arbitrarily on the
entire target language. The following lemma supplies that step.

\begin{lemma}[Bounded capture with core avoidance]\label{lem:capture}
Let $(U_n)_{n\ge1}$ be finite subsets of a set $Y$ with
$\sup_n|U_n|<\infty$, and let $\mathcal C$ be an at-most-countable family
of infinite subsets of $Y$. There is $D\subseteq Y$ such that
$U_n\subseteq D$ for infinitely many $n$, but $C\not\subseteq D$ for
every $C\in\mathcal C$.
\end{lemma}
\begin{proof}
Choose $d$ with $|U_n|\le d$. If $\mathcal C$ is empty, take
$D=\bigcup_n U_n$. Otherwise list its members as $C_1,C_2,\ldots$,
repeating a finite list if necessary.

Maintain an infinite set $I_{j-1}$ of available indices, starting with
$I_0=\N$. At step $j$, choose $d+1$ points of
$C_j\setminus\bigcup_{i<j}U_{n_i}$; this is possible because the removed
union is finite. Each $U_n$ contains at most $d$ of these points.
By the infinite pigeonhole principle, some chosen point $x_j$ is omitted
by infinitely many $U_n$ with $n\in I_{j-1}$. Let
\[
 I_j=\{n\in I_{j-1}:x_j\notin U_n\},
\]
and choose $n_j\in I_j$ larger than every previously selected index.

Set $D=\bigcup_jU_{n_j}$. The indices $n_j$ are distinct, so $D$ captures
infinitely many terms. For each $j$, the choice of $x_j$ excludes it from
every earlier selected set; the nested pools exclude it from the current
and every later selected set. Hence $x_j\in C_j\setminus D$ for every $j$.
\end{proof}

\subsection{Every finite level}

Fix $k\ge1$. Choose four pairwise disjoint sets: anchors
$A_0=\{a_0,\ldots,a_{k-1}\}$, $B_0=\{b_0,\ldots,b_{k-1}\}$
and countably infinite tails $A_*,B_*$. Put
$A=A_0\cup A_*$, $B=B_0\cup B_*$ and $\X=A\sqcup B$.
Define the anchored class
\begin{align}
 L_i(D)&=A\cup(B_0\setminus\{b_i\})\cup D,
       & i<k,\quad D\subseteq B_*,\label{eq:anchor-left}\\*
 K_j(E)&=B\cup(A_0\setminus\{a_j\})\cup E,
       & j<k,\quad E\subseteq A_*,\label{eq:anchor-right}\\*
 \cH^{(k)}&=\{L_i(D):i<k,D\subseteq B_*\}
        \cup\{K_j(E):j<k,E\subseteq A_*\}.\notag
\end{align}
Each side has an infinite common block. Opposite-side types omit different
specified anchors, while their tails vary without restriction.

For $k=3$, there are nine cross-side type pairs. Each can be separated by
putting $a_j$ in the left witness or $b_i$ in the right witness.
Figure~\ref{fig:finite-hierarchy} gives an assignment using two anchors per
left witness and one per right witness. Each vertex represents a whole
type, with every tail choice retained.

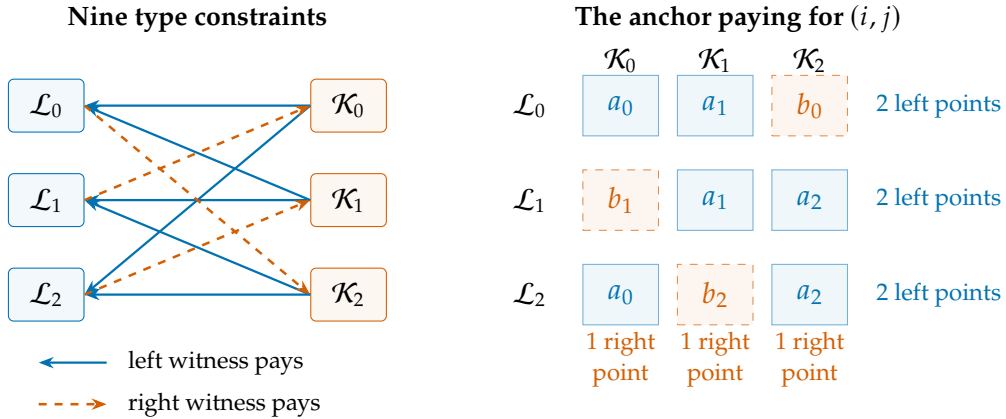
\begin{figure}[htbp]
\centering
\begingroup
\definecolor{LeftBlue}{HTML}{0072B2}
\definecolor{RightRed}{HTML}{D55E00}
\begin{tikzpicture}[x=1cm,y=1cm,
  vertex/.style={draw,rounded corners=2pt,minimum width=10mm,minimum height=7mm,font=\small},
  leftedge/.style={-{Stealth[length=2mm]},LeftBlue,line width=.85pt},
  rightedge/.style={-{Stealth[length=2mm]},RightRed,dashed,line width=.95pt}]
\node[font=\small\bfseries] at (2.0,3.65) {Nine type constraints};
\node[font=\small\bfseries] at (9.15,3.65) {The anchor paying for $(i,j)$};

\foreach \i/\y in {0/2.5,1/1.25,2/0} {
  \node[vertex,draw=LeftBlue,fill=LeftBlue!5] (L\i) at (0,\y) {$\mathcal L_{\i}$};
  \node[vertex,draw=RightRed,fill=RightRed!5] (K\i) at (4,\y) {$\mathcal K_{\i}$};
}
\draw[leftedge] (K0.west)--(L0.east);
\draw[leftedge] (K1.west)--(L0.east);
\draw[leftedge] (K1.west)--(L1.east);
\draw[leftedge] (K2.west)--(L1.east);
\draw[leftedge] (K2.west)--(L2.east);
\draw[leftedge] (K0.west)--(L2.east);
\draw[rightedge] (L0.east)--(K2.west);
\draw[rightedge] (L1.east)--(K0.west);
\draw[rightedge] (L2.east)--(K1.west);

\foreach \j/\x in {0/7.6,1/8.85,2/10.1} {
 \node[font=\small] at (\x,3.1) {$\mathcal K_{\j}$};
}
\foreach \i/\y in {0/2.5,1/1.25,2/0} {
 \node[font=\small,anchor=east] at (6.75,\y) {$\mathcal L_{\i}$};
 \node[font=\footnotesize,LeftBlue,anchor=west] at (10.85,\y) {2 left points};
}
\foreach \x/\y/\lab in {7.6/2.5/a_0,8.85/2.5/a_1,8.85/1.25/a_1,10.1/1.25/a_2,7.6/0/a_0,10.1/0/a_2} {
 \node[draw=LeftBlue!55,fill=LeftBlue!6,minimum width=10mm,minimum height=8mm,
 text=LeftBlue] at (\x,\y) {$\lab$};
}
\foreach \x/\y/\lab in {10.1/2.5/b_0,7.6/1.25/b_1,8.85/0/b_2} {
 \node[draw=RightRed!70,dashed,fill=RightRed!6,minimum width=10mm,minimum height=8mm,
 text=RightRed] at (\x,\y) {$\lab$};
}
\foreach \x in {7.6,8.85,10.1} {
 \node[font=\footnotesize,RightRed,align=center] at (\x,-.85) {1 right\\point};
}
\draw[leftedge] (.8,-.9)--(-.1,-.9);
\node[font=\footnotesize,anchor=west] at (.95,-.9) {left witness pays};
\draw[rightedge] (-.1,-1.45)--(.8,-1.45);
\node[font=\footnotesize,anchor=west] at (.95,-1.45) {right witness pays};
\end{tikzpicture}
\endgroup
\caption{\textbf{The width-two construction at $k=3$.}
Write $\mathcal L_i=\{L_i(D):D\subseteq B_*\}$ and
$\mathcal K_j=\{K_j(E):E\subseteq A_*\}$.
An arrow points toward the endpoint supplying the separating anchor:
solid blue edges use $a_j$ in the left witness, and dashed orange edges
use $b_i$ in the right witness. The matrix lists the selected anchor for
each pair. Every left witness uses two points; every right witness uses
one. This gives an upper bound of two. The lower bound must first show
that arbitrary tail-dependent witnesses cannot evade the anchor constraints.}
\label{fig:finite-hierarchy}
\end{figure}

\begin{proposition}[Exact finite levels]\label{thm:finite-levels}
For every $k\ge1$,
\[
 \spwidth(\cH^{(k)})=\left\lceil\frac{k}{2}\right\rceil.
\]
Consequently every positive integer is attained by a class generatable in the limit.
\end{proposition}
\begin{proof}
\emph{Upper bound.}
Put $d=\lceil k/2\rceil$. For each row $i$, let $J_i$ be the $d$ cyclic
columns $i,i+1,\ldots,i+d-1$ modulo $k$. Set
\[
 T(L_i(D))=\{a_j:j\in J_i\},\qquad
 T(K_j(E))=\{b_i:j\notin J_i\}.
\]
These positive witnesses have sizes $d$ and $k-d\le d$, respectively,
because each column occurs in exactly $d$ of the sets $J_i$.
Every pair $(i,j)$ is separated: if $j\in J_i$, the first witness contains
$a_j$ omitted by $K_j(E)$; otherwise the second contains $b_i$ omitted by
$L_i(D)$. Hence no opposite-side pair can be simultaneously active.
A nonempty active family shares all $A$ or all $B$, so its intersection is
infinite.

\emph{Lower bound.}
The finite graph alone is not a lower bound: a witness may use tail points
and depend on the entire target. We first choose targets that force all
$k^2$ pairs to use the displayed anchor incidences.
Fix $q\in\N_0$ and an arbitrary valid assignment $T$ with
$|T(L)|\le q$ for every target.
Enumerate $A_*$ and let $E_n$ be its first $n$ points.
For $j<k$, put $K_{j,n}=K_j(E_n)$. Pass to increasing indices $n$ on
which all $k$ anchor sets
\[
 R_j=T(K_{j,n})\cap B_0
\]
are fixed simultaneously. This is possible because their joint pattern
has only finitely many values. On those indices set
\[
 U_n=\bigcup_{j<k}\bigl(T(K_{j,n})\cap B_*\bigr),
 \qquad |U_n|\le kq.
\]
For every $i<k$ and finite sample $S$, form
\begin{align*}
 \mathcal F_i(S)&=\{L_i(D'):T(L_i(D'))\subseteq S\subseteq L_i(D')\},\\
 C_{i,S}&=\bigcap\{D'\subseteq B_*:L_i(D')\in\mathcal F_i(S)\},
\end{align*}
with empty intersection $B_*$. There are only countably many pairs
$(i,S)$. Apply \cref{lem:capture} to the infinite members of this
collection of cores and to $(U_n)$. Obtain $D\subseteq B_*$ capturing
infinitely many $U_n$ and containing none of those infinite cores.
The first property accommodates infinitely many right-witness tails.
The second prevents that accommodation from containing an infinite
active core on the left.

Fix the $k$ targets $L_i(D)$. We claim that every pair $(i,j)$ obeys
\begin{equation}\label{eq:anchor-cover}
 a_j\in T(L_i(D))\quad\text{or}\quad b_i\in R_j.
\end{equation}
If both fail, choose a captured $n$ large enough that
$T(L_i(D))\cap A_*\subseteq E_n$. The target $K_{j,n}$ contains all $B$
and every $A$-anchor except $a_j$, so
$T(L_i(D))\subseteq K_{j,n}$. Conversely $L_i(D)$ contains all $A$;
the $B$-anchor part $R_j$ omits $b_i$, and the $B$-tail part of
$T(K_{j,n})$ lies in $U_n\subseteq D$. Thus
$T(K_{j,n})\subseteq L_i(D)$.

Both targets are active at
$S=T(L_i(D))\cup T(K_{j,n})$. The full active intersection $J$ is infinite.
Since $K_{j,n}$ has only finitely many $A$-points and $B_0$ is finite,
$J\cap B_*$ is infinite. Every target in $\mathcal F_i(S)$ contains $J$,
so $J\cap B_*\subseteq C_{i,S}$. This core is infinite, yet
$C_{i,S}\subseteq D$ because $L_i(D)$ itself is active. This contradicts
the choice of $D$ and proves \eqref{eq:anchor-cover}.

The $k^2$ pairs are therefore covered by at most
\[
 \sum_{i<k}|T(L_i(D))\cap A_0|+\sum_{j<k}|R_j|\le 2kq
\]
anchor incidences. It follows that $q\ge\lceil k/2\rceil$.
The counting step applies to arbitrary tail-dependent assignments because
the capture argument first established \eqref{eq:anchor-cover}; it does not
assume witnesses use anchors only. The definition of width and \cref{prop:separation} give the exact width.
Taking $k=2r$ realizes each positive integer $r$.
\end{proof}

The upper bound assigns each edge of a complete bipartite graph to an
endpoint; capture makes the corresponding incidence budget unavoidable
in the lower bound. For $k=3$, that budget reads $9\le6q$, forcing $q\ge2$.

The full types in Figure~\ref{fig:finite-hierarchy} are essential.
If we retain only the six empty-tail representatives
$L_i(\varnothing),K_j(\varnothing)$, choose $a_*\in A_*$ and $b_*\in B_*$.
Assigning $\{a_*\}$ to every left representative and $\{b_*\}$ to every
right representative separates every cross pair. This six-language
family has width one: these witnesses give the upper bound, and its
empty full intersection rules out width zero. A single tail point now
handles all three constraints incident to a vertex. The full language
class lets the capture argument eliminate this shortcut before counting.

\subsection{The unbounded-finite level}

Let $\X=A\sqcup B$ with both blocks countably infinite, and define
\begin{equation}\label{eq:two-core}
 \cH_{\mathrm{two}}=\{A\cup D:D\subseteq B\}
                   \cup\{B\cup E:E\subseteq A\}.
\end{equation}
\begin{proposition}[The $\omega$ level]\label{thm:omega}
The class $\cH_{\mathrm{two}}$ is generatable in the limit and
$\spwidth(\cH_{\mathrm{two}})=\omega$.
\end{proposition}
\begin{proof}
Enumerate the blocks as $(a_i)_{i\ge0}$ and $(b_i)_{i\ge0}$.
For a proper target $A\cup D$, put
$i=\min\{r:b_r\notin D\}$ and assign $\{a_0,\ldots,a_i\}$.
For a proper target $B\cup E$, put
$j=\min\{r:a_r\notin E\}$ and assign $\{b_0,\ldots,b_j\}$.
Assign the empty witness to $\X$, the only target on both sides.
Opposite proper targets cannot be active together: containment of the
first witness in the second target forces $j>i$, whereas containment of
the second witness in the first forces $i>j$.
Every nonempty active family therefore shares an infinite block, or
consists only of $\X$. This gives width at most $\omega$.

For any $k$, use the first $k$ points in each block as anchors. The
corresponding $\cH^{(k)}$ is a subfamily of $\cH_{\mathrm{two}}$.
Monotonicity and \cref{thm:finite-levels} give width at least
$\lceil k/2\rceil$ for every $k$. It cannot be finite and must equal
$\omega$.
\end{proof}

Thus the pointwise finiteness in \cref{thm:characterization} cannot be
strengthened to a uniform finite cardinality bound. A stronger fact holds:
along the explicit chain $B\cup\{a_0,\ldots,a_{n-1}\}$, the number of
$B$-points in every valid assignment's witness tends to infinity.
\Cref{prop:chain-divergence} proves this assertion.

\Needspace{11\baselineskip}
\subsection{The endpoints and full range}

\begin{proposition}\label{prop:endpoints}
The family of all cofinite subsets of $\N$ has width $1$, whereas the
family $\infinite{\N}$ of all infinite subsets has width $\omega+1$.
\end{proposition}
\begin{proof}
The cofinite family is countable and has empty common intersection, so
\cref{thm:singleton,prop:width-basic} give width exactly one.
For all infinite targets, fix an arbitrary total generator. Maintain
finite disjoint sets of revealed and permanently forbidden points.
At each round reveal a new point outside their union. If the output is
already revealed, freshness fails; otherwise permanently forbid it.
There is always another point available. The resulting input sequence
has infinite range $L$ and is exhaustive for that fixed resulting target.
Every output either repeats an input or is absent from $L$. Thus the
generator fails on $L$. Since the generator was arbitrary,
\cref{thm:width} gives width $\omega+1$.
\end{proof}

\begin{proof}[Completion of \cref{thm:hierarchy}]
The countable bound is \cref{thm:singleton}. A single infinite target has
width zero. Taking $k=2d$ in \cref{thm:finite-levels} realizes every positive
integer $d$, and \cref{thm:omega} realizes $\omega$. These classes are
generatable by \cref{cor:width-characterization}.
The all-infinite family has width $\omega+1$ by \cref{prop:endpoints}.
Bijective relabeling gives the claims on any countably infinite $\X$.
\end{proof}

\begin{table}[htbp]
\centering
\begin{tabular}{@{}lll@{}}
\toprule
Class & Width & Reason\\
\midrule
Infinite common core & $0$ & Empty witnesses\\
Any countable class & $\le1$ & Distinct singleton witnesses\\
Anchored class $\cH^{(k)}$ & $\lceil k/2\rceil$ & Matching incidence bounds\\
Two-core class $\cH_{\mathrm{two}}$ & $\omega$ & Unbounded finite witnesses\\
All infinite subsets & $\omega+1$ & No finite witness assignment\\
\bottomrule
\end{tabular}
\caption{The separation-width hierarchy. The two infinite levels distinguish
unbounded finite certificates from the absence of any finite-certificate
assignment. These widths are not convergence-time bounds.}
\label{tab:widths}
\end{table}

\section{Why local dimensions cannot suffice}
\label{sec:barriers}

The preceding width is global: its assignment must handle every bad
subfamily at once. The following obstruction explains why a dimension
built only from locally realized configurations cannot simply replace it.

\begin{theorem}[Countable-support obstruction]\label{thm:countable-barrier}
There is no invariant $D(\cH)\in\N_0\cup\{\infty\}$ satisfying both
\begin{enumerate}[label=\textup{(\roman*)}]
 \item $\cH$ is generatable in the limit if and only if $D(\cH)<\infty$;
 \item whenever $D(\cH)=\infty$, some countable
       $\cH_0\subseteq\cH$ also has $D(\cH_0)=\infty$.
\end{enumerate}
No monotonicity hypothesis is needed for this impossibility.
\end{theorem}
\begin{proof}
All infinite subsets of $\N$ form a nongeneratable class by
\cref{prop:endpoints}, so (i) gives infinite dimension. Condition (ii)
then supplies a countable subfamily of infinite dimension. That subfamily
is generatable by \cref{thm:singleton,thm:characterization}, contradicting
(i).
\end{proof}

For a dimension defined by configurations of arbitrarily large finite
depth, condition (ii) follows whenever each depth has at most countably
many language realizers and collecting those realizers preserves the
configuration. Choose one witness family at each integer depth and take
their countable union. Finite VC shattering and finite-depth Littlestone
shattering have this property. A finite positive-closure obstruction also
has countable support: select, for each point outside its finite core, a
consistent language omitting that point. Adding more consistent languages
can only shrink the core.

There is also a finite-level failure of countable determination. For every
integer $d\ge1$, the class $\cH^{(2d+1)}$ has width $d+1$, while each of
its countable subfamilies has width at most one. Thus even the obstruction
$\spwidth(\cH)>d$ need not be witnessed by a countable subclass.

This argument does not cover every infinitary rank. A tree condition
requiring each complete infinite branch to be realized by an actual
target, or a global compatibility condition such as
\eqref{eq:separation}, need not be supported on a countable subfamily.
The theorem's explicit support assumption is essential.

\Needspace{5\baselineskip}
\begin{proposition}[Identical finite profiles]\label{prop:profiles}
The cofinite class $\cH_{\mathrm{cf}}$ and the class
$\cH_{\mathrm{inf}}=\infinite{\N}$ have identical traces on every finite
domain and identical positive closures at every finite sample. Their
generatability in the limit differs.
\end{proposition}
\begin{proof}
For finite $F\subseteq\N$ and $E\subseteq F$, the cofinite target
$\N\setminus(F\setminus E)$ has trace $E$ on $F$. Thus both classes
realize all finite traces. At a finite positive sample $S$, every
consistent target contains $S$, whereas for each $x\notin S$ the
cofinite target $\N\setminus\{x\}$ is consistent and omits $x$.
Both positive closures are exactly $S$. Their differing status follows
from \cref{prop:endpoints,thm:width}.
\end{proof}

Consequently, even the complete collection of finite trace sets or of
finite-sample positive closures cannot determine generation in the limit.
This includes invariants determined solely by these profiles and their
finite iterations; it does not include enriched data that retain
additional relations among the realizing languages.

The width measures the sizes of the witnesses themselves. An alternative
based on the size of bad active samples collapses to zero or infinity;
Appendix~\ref{app:padding} gives the short padding argument.

\section{Scope and computational interpretation}
\label{sec:discussion}

The characterization separates two existence questions.
A successful generator yields a finite witness for every target, and
compatible witnessed targets must have infinitely many common elements.
Conversely, an assignment with this intersection property specifies a
generator. This is a structural existence statement about the family.
The assignment and the intersections in
\eqref{eq:generator-from-witnesses} need not be computable.

The normalization theorem has a separate effective content: given a total
computable successful generator, \cref{alg:normalization} yields a total
computable set-driven replacement. It requires neither effective membership
in the unknown target nor an effective procedure for extracting $T_L$.
The search may examine many words and is not accompanied by an efficiency
guarantee. Set-drivenness removes dependence on order and multiplicity; it
does not bound the memory required to retain the observed set.
Memory constraints define a separate restriction on generation, studied by
\citet{kleinberg2026memory}.

For a single infinite target $L\subseteq\N_0$, the effective locking
conclusion yields an exact criterion: a total computable generator exists
if and only if $L$ contains an infinite computably enumerable subset
(Appendix~\ref{app:computability}). In particular, there is a target of
asymptotic density one with width zero and no computable generator.
Moreover, the countable class of all decidable
density-one targets has width one and individually computable generators,
yet no common total computable generator. These examples make precise the
computational information that witness size alone does not supply.

The assumptions also have distinct roles in the proof. Countability supplies
an exhaustive ordering of target elements and a finite-predecessor ordering
of words. Infinitude supplies fresh outputs and gives checkpoints of every
finite size. Exhaustiveness both rules out an infinite marked error sequence
and ensures that the final finite witness is eventually observed.
The theorem covers precisely these positive-text, single-element generation
requirements; it imposes no identification or output-diversity requirement.

\paragraph{Combinatorial meaning and remaining questions.}
Positive separation width measures the number of positive incidences that
each target must contribute to separate every subfamily with finite core.
The finite anchored examples give matching lower and upper bounds, and the
$\omega$ example separates pointwise finiteness from every uniform finite
budget. The definition still optimizes a global assignment; the present
results do not provide an intrinsic tree-rank formula or a method to compute
the width from a presentation of the class. The countable-support obstruction
identifies a constraint on such a formula rather than prohibiting all
infinitary alternatives.

It remains useful to ask which additional regularity assumptions permit
canonical witness choices, which operations preserve finite width, and what
extra information would relate witness size to observation or mistake
complexity. Exhaustive texts alone allow any finite certificate to be delayed,
so such quantitative conclusions require separate hypotheses. Efficient
witness discovery and active-intersection selection remain separate algorithmic
questions.

\paragraph{Formal verification.}
The characterization and full width hierarchy are checked in Lean. The
simplified normalization has a separately checked implementation, and the
diagonal capture argument has its own checked proof. Appendix~\ref{app:formal}
describes the correspondence, implementation variants, and trusted dependencies.
The additional density-one counterexample and computability results in
Appendices~\ref{app:quantifiers} and~\ref{app:computability} have written
proofs and are outside the current Lean development.

\paragraph{AI Disclosure.}
OpenAI Codex assisted with mathematical exploration, proof development,
Lean formalization, and manuscript preparation. Responsibility for the
paper's claims and interpretations rests with the authors.

\phantomsection
\addcontentsline{toc}{section}{References}
\bibliographystyle{plainnat}   
\bibliography{refs}

\appendix
\section{From sorting failure to successful normalization}
\label{app:sorting}
\label{app:worked-normalization}

We follow one generator through the obstruction to sorting, the internal
rounds of \cref{alg:normalization}, and the confirming witness in its proof.
The same example shows why a valid but unobserved output can delay the
simulation, and why later simulated histories may keep changing after
correctness has been secured.

\subsection{Why sorting fails}
\label{app:sorting-failure}

The normalization theorem constructs a new generator by finite simulation.
The following example shows why simply sorting the input to an existing
successful generator is not a valid replacement.

Let $\X=\N_0$ and let the sole target be $L=\N$. Set $G(\ew)=1$.
For a nonempty finite word $\sigma$, define
\[
 G(\sigma)=
 \begin{cases}
  0,&\sigma\text{ is strictly increasing and }
       \cont(\sigma)\ne\{1,\ldots,\max\cont(\sigma)\},\\
  \max\cont(\sigma)+1,&\text{otherwise}.
 \end{cases}
\]
On every text for $L$, this generator is eventually correct and fresh.
If the entire text is strictly increasing, exhaustiveness forces it to be
$1,2,3,\ldots$, and every output is correct.
If it is not strictly increasing, there is a finite prefix that witnesses
this fact. All later prefixes retain that witness, and their outputs are
the maximum observed value plus one, hence valid and fresh.

Now define $h(S)=G(\operatorname{sort}(S))$, where $\operatorname{sort}(S)$
lists the distinct elements of $S$ in increasing order. On the text
\[
  1,3,2,5,4,7,6,\ldots,
\]
every prefix ending at a newly introduced odd number $2k+1$, for $k\ge1$,
omits $2k$. Sorting that prefix therefore triggers the output $0$.
The generator $h$ makes infinitely many errors.

The example separates these two transformations of a particular generator;
it does not separate sequence-input and set-input generatability.
The singleton family plainly has a successful set-driven generator.

\subsection{Tracing the normalization}
\label{app:normalization-trace}

Use the usual order $0<1<2<\cdots$ on $\X=\N_0$. To make every
least-code choice explicit, order finite words first by their weight
\[
 \operatorname{wt}(\sigma)=\sum_{i=1}^{|\sigma|}(\sigma_i+1),
 \qquad \operatorname{wt}(\ew)=0,
\]
and then lexicographically among words of equal weight. Define
$\code(\sigma)$ to be its zero-based position in this order. Each weight
contains finitely many words: a word of weight $r$ has length at most $r$
and entries at most $r-1$. Thus this is an effective numbering with finite
predecessor sets, as required in \cref{sec:construction}. The weights in
the table below are ordering aids, not the codes themselves; for example,
$\code((1))=3$ and $\code((1,3))=55$.

Let $F$ be the fresh repair in \eqref{eq:fresh-repair}. On words over the
target $L=\N$, the generator $G$ is already fresh: it outputs either $0$
or the maximum input plus one. Hence $F=G$ on every history searched below.
The repair still matters on some words containing $0$, which will enter
the general confirming-witness formula.

Every call with a sample $S$ starts a new simulation at $\ew$ and performs
$|S|$ rounds. The rows of \cref{tab:normalization-trace} are these internal
rounds, not arrival times in an input text. Each block holds $S$ fixed.

\begin{table}[H]
\centering
\small
\renewcommand{\arraystretch}{1.15}
\begin{tabular}{@{}ccclcc@{}}
\toprule
Sample $S$ & Round $k$ & $E_k(S)$ & Selected history $\tau$
 & $\operatorname{wt}(\tau)$ & $F(\tau)$ \\
\midrule
$\{1,3\}$ & 1 & $\{1\}$ & $(1)$ & 2 & 2 \\
 & 2 & $\{1,3\}$ & $(1,3)$ & 6 & 0 \\
\midrule
$\{1,2,3\}$ & 1 & $\{1\}$ & $(1,3)$ & 6 & 0 \\
 & 2 & $\{1,2\}$ & $(1,3,2)$ & 9 & 4 \\
 & 3 & $\{1,2,3\}$ & $(1,3,2,1)$ & 11 & 4 \\
\midrule
$\{1,2,3,4\}$ & 1 & $\{1\}$ & $(1,3)$ & 6 & 0 \\
 & 2 & $\{1,2\}$ & $(1,3,2,4)$ & 14 & 5 \\
 & 3 & $\{1,2,3\}$ & $(1,3,2,4,1)$ & 16 & 5 \\
 & 4 & $\{1,2,3,4\}$ & $(1,3,2,4,1,1)$ & 18 & 5 \\
\bottomrule
\end{tabular}
\caption{Three complete calls to the normalized generator. Each selected
history strictly extends the previous one in its block, contains the
required checkpoint, and has output outside $S$. The last output in each
block is $g(S)$: respectively $0$, $4$, and $5$.}
\label{tab:normalization-trace}
\end{table}

For $S=\{1,3\}$, the first candidate is $(1)$, whose output $2$ is valid
but unobserved. The search tests only whether an output has been seen;
it therefore accepts this candidate. At the second round, the checkpoint
requires both $1$ and $3$, making $(1,3)$ the least eligible extension.
Its output is $0$, so $g(\{1,3\})=0$. A normalized generator can still
err before a locking witness has appeared.

For $S=\{1,2,3\}$, the output $2$ has been confirmed. More precisely,
the complete list of positive words containing $1$ that precede $(1,3)$,
together with their outputs, is
\[
\begin{array}{c|ccccc}
 \tau &(1)&(1,1)&(1,2)&(2,1)&(1,1,1)\\
 F(\tau)&2&2&3&3&2
\end{array}
\]
All are rejected because their outputs belong to $S$. The first selected
history is therefore $(1,3)$, with output $0\notin S$. The second round
must append $2$ to meet $E_2(S)$; the least extension is $(1,3,2)$.
It is no longer increasing, so its output is $4$. In the third round
all sample points are already present in the history, but a strict
extension is still required. Appending the smallest sample point, $1$,
gives $(1,3,2,1)$ and output $4$. Thus $g(\{1,2,3\})=4$.

For $S=\{1,2,3,4\}$, the first selected history remains $(1,3)$.
The former second-round choice $(1,3,2)$ is now rejected, since its
output $4$ has been observed. Any extension meeting $E_2(S)=\{1,2\}$
places $2$ after $3$ and hence is not increasing. To have output outside
$S$, its maximum must therefore be $4$. The least such extension appends
$(2,4)$; it has smaller lexicographic order than the equally weighted
extension appending $(4,2)$. Subsequent rounds append $1$, giving the
last two rows and $g(\{1,2,3,4\})=5$. Confirming another valid output
has changed the simulated history while preserving correctness.

\subsection{The error chain and its witness}
\label{app:normalization-witness}

The proof-only genuine-error chain is particularly short:
\[
 \sigma_0=\ew,\qquad \sigma_1=(1,3),\qquad m=1.
\]
Indeed, $(1,3)$ is the least word over $L$ containing $E_1(L)=\{1\}$
with output outside $L$; the preceding list accounts for all earlier
possibilities. Every strict extension of $(1,3)$ that contains
$E_2(L)=\{1,2\}$ places $2$ after $3$, so its output is the maximum
entry plus one and belongs to $L$. No second genuine error exists.
This is exactly the terminal-safety statement
\eqref{eq:terminal-safety} for this example.

We can also evaluate the full witness formula
\eqref{eq:normalization-witness}. Its confirming outputs are
\[
 \{F(\tau):\code(\tau)<\code((1,3)),\ F(\tau)\in L\}
   =\{1,2,3,4\}.
\]
For completeness, all earlier words have weight at most $6$. Among words
up to this weight, the only one with repaired output greater than $4$
is $(4,0)$, which is later than $(1,3)$. Conversely, $\ew$, $(1)$,
$(1,2)$, and $(3,0)$ are earlier and have outputs $1,2,3,4$,
respectively. Consequently the witness supplied by that formula is
\[
 T_L=\underbrace{\{1,2\}}_{E_2(L)}
       \cup\underbrace{\{1,3\}}_{\cont(\sigma_1)}
       \cup\underbrace{\{1,2,3,4\}}_{\text{confirming outputs}}
     =\{1,2,3,4\}.
\]
The extra confirmation $4$ arises from histories such as $(3,0)$, which
cannot occur over this target. The general formula includes them to avoid
restricting each predecessor set to legal histories. This explains why
the formula supplies a sufficient certificate without claiming a smallest
one.

In fact, a direct argument gives the stronger statement
\begin{equation}\label{eq:worked-locking}
 \{1,2,3\}\subseteq S\subseteq\N,\quad S\text{ finite}
 \quad\Longrightarrow\quad g(S)=\max S+1.
\end{equation}
The first round chooses $(1,3)$, because all five earlier possibilities
have their outputs $2$ or $3$ in $S$. The second round must insert $2$
after $3$, so this and all later histories are not increasing. The final
round includes $E_{|S|}(S)=S$; since the history contains only points of
$S$, its content is exactly $S$. Its repaired output is therefore
$\max S+1\in\N\setminus S$. This proves locking on every finite
extension of the smaller witness $\{1,2,3\}$, not just the samples in
the table.

Thus the example realizes the proof mechanism: positive confirmations
force agreement through the sole genuine error, and a further round
passes into the region of terminal safety. The search never asks whether
an output belongs to $L$. The weights, code choices, and traces here are
fixed for exposition; they are not execution traces of the different
word encoding used by the repository's Lean implementation.
Finally, these are locking witnesses for this particular normalized
generator. The singleton class $\{\N\}$ itself has separation width
zero, since its common intersection is infinite.

\section{Witness divergence along an explicit chain}
\label{app:divergence}

The width-$\omega$ example has a stronger property than the absence of a
uniform bound. Its witnesses must become large along one specified chain,
regardless of how the assignment depends on the targets.

\begin{proposition}\label{prop:chain-divergence}
Let $T$ be any valid finite-witness assignment for
$\cH_{\mathrm{two}}$ in \eqref{eq:two-core}. Enumerate $A$ as
$(a_i)_{i\ge0}$ and put
\[
 E_n=\{a_0,\ldots,a_{n-1}\},\qquad K_n=B\cup E_n.
\]
Then $|T(K_n)\cap B|\to\infty$ as $n\to\infty$.
The symmetric assertion holds with $A$ and $B$ interchanged.
\end{proposition}
\begin{proof}
Suppose some infinite subsequence of
$Q_n=T(K_n)\cap B$ has bounded cardinality. For each finite sample $S$,
let $\mathcal F_A(S)$ consist of its active targets of the form $A\cup D'$,
and set
\[
 C_S=\bigcap\{D'\subseteq B:A\cup D'\in\mathcal F_A(S)\},
\]
with empty intersection $B$. There are only countably many such cores.
Apply \cref{lem:capture} to the bounded subsequence and to the infinite
members of this core collection. It supplies $D\subseteq B$ capturing
infinitely many $Q_n$ and containing no infinite $C_S$.

For the target $L=A\cup D$, the finite set $T(L)\cap A$ lies in $E_n$
for all sufficiently large $n$. Choose such a captured index. Positivity
gives $T(L)\subseteq K_n$ and $T(K_n)\subseteq L$, so both targets are
active at $S=T(L)\cup T(K_n)$. Their full active intersection $J$ is
infinite. Since $J\subseteq B\cup E_n$ and $E_n$ is finite,
$J\cap B$ is infinite. It is contained in $C_S$, while the active target
$L$ gives $C_S\subseteq D$. This contradicts the construction of $D$.
Thus no bounded-cardinality infinite subsequence exists, which is exactly
the asserted divergence. The symmetric proof interchanges the blocks.
\end{proof}

This argument gives no universal rate of divergence and no generation-time
bound. It concerns the cardinalities of the certificates assigned to a
specific increasing sequence of languages.

\section{Why optimizing bad-sample size collapses}
\label{app:padding}

A different shortcut fails for a quantitative reason. One might assign
finite witnesses first, measure the largest bad active sample, and then
optimize that measurement. The optimization removes every finite defect.

\begin{proposition}[Padding collapse]\label{prop:padding-collapse}
For a finite positive assignment $T$, define
\[
 b_T(\cH)=\sup\bigl\{|S|+1:\activefamily(S)\ne\varnothing,
               \ |\Core(\activefamily(S))|<\infty\bigr\},
\]
with empty supremum zero and unbounded supremum $\infty$. Then
\[
 \inf_T b_T(\cH)=
 \begin{cases}0,&\cH\text{ generatable},\\
 \infty,&\cH\text{ nongeneratable}.
 \end{cases}
\]
\end{proposition}
\begin{proof}
Suppose $b_T(\cH)=d<\infty$. Enlarge each $T(L)$ within its infinite
target to a finite $T'(L)$ of size at least $d$, retaining $T(L)$.
If a sample $S$ activates a target under $T'$, then $|S|\ge d$ and the
original active family is nonempty. Its core cannot be finite, since
that would imply $|S|+1\le d$. The new active family is smaller, so its
core contains the old infinite core. Thus $T'$ is valid and
$b_{T'}(\cH)=0$.
A finite infimum over $\N_0\cup\{\infty\}$ entails some finite-valued
assignment and hence, by this argument, a zero-valued one. Conversely a
valid assignment has value zero. Apply \cref{thm:characterization}.
\end{proof}

Positive separation width retains the size cost of padding and therefore
escapes this collapse. \Cref{thm:finite-levels,thm:omega} show that this
cost has every finite level as well as a necessary unbounded-finite level.

\section{Why the quantifiers matter}
\label{app:quantifiers}

The finite-witness condition fixes one assignment before testing any
sample. Its force comes from compatibility across the whole active
family. A single counterexample shows why neither finite intersections
nor separate assignments on countable subfamilies can replace that
requirement.

\paragraph{Finite intersections versus the full intersection.}
Write $[0,n)=\{0,\ldots,n-1\}$ and let
\[
 \mathcal D_1=\left\{L\subseteq\N_0:
 \lim_{n\to\infty}\frac{|L\cap[0,n)|}{n}=1\right\}
\]
be the family of languages of asymptotic density one. Every finite
intersection of its members still has density one, since its complement
is a finite union of density-zero sets.

\begin{proposition}[Abundant finite intersections do not suffice]
\label[proposition]{prop:density-one}
The class $\mathcal D_1$ is not generatable in the limit. For the empty
witness assignment, every nonempty finite subfamily of every active
family has infinite intersection, but the full active intersection at
each finite sample $S$ is exactly $S$.
\end{proposition}
\begin{proof}
Every active target contains $S$. For each $y\notin S$, the cofinite
target $\N_0\setminus\{y\}$ is active and omits $y$. This proves the
intersection assertion.

Fix any total generator $G$. Starting with $\sigma_0=\ew$ and
$Q_0=\varnothing$, build nested finite histories and forbidden sets.
At stage $r\ge1$, append in increasing order every previously unseen
point of $[0,2^r]\setminus Q_{r-1}$, obtaining $\sigma_r$.
If $y_r=G(\sigma_r)$ has already appeared, put $Q_r=Q_{r-1}$;
otherwise put $Q_r=Q_{r-1}\cup\{y_r\}$.
Observed and forbidden points remain disjoint. Every newly forbidden
point exceeds $2^r$, since all smaller nonforbidden points have just
been shown, and at most one point is newly forbidden at each stage.
For $Q=\bigcup_rQ_r$ and $n\ge2$,
\[
 |Q\cap[0,n)|\le |\{r\ge1:2^r<n\}|\le\log_2 n.
\]
Thus $L=\N_0\setminus Q$ has density one.

The first block is $0,1,2$. For $r\ge2$, the new interval
$(2^{r-1},2^r]$ contains $2^{r-1}>r-1$ points, of which at most
$r-1$ are forbidden, so every block is nonempty. Every point of $L$
eventually appears. The concatenated blocks therefore form a text for
$L$, with strictly increasing checkpoint times $|\sigma_r|$.
At each checkpoint the output is either already observed or belongs to
$Q$, so $G$ fails infinitely often on this text.
\end{proof}

By \cref{thm:characterization}, no finite-witness assignment can repair
this class. Thus finite intersections fail even as a replacement
existence criterion, not merely as a test of the empty assignment.

\paragraph{Separate local assignments versus one global assignment.}
For a \emph{fixed} assignment, countable bad subfamilies suffice to test
separation by \cref{prop:separation}. But the implication
\[
 \bigl[\forall\text{ countable }\mathcal K\subseteq\cH\
       \exists T_{\mathcal K}\text{ valid on }\mathcal K\bigr]
 \quad\Longrightarrow\quad
 \bigl[\exists T\text{ valid on }\cH\bigr]
\]
is false. Take $\cH=\mathcal D_1$: every countable subfamily admits
singleton witnesses by \cref{thm:singleton}, whereas
\cref{prop:density-one} rules out a global finite assignment. Even a
size-one bound on all the separate choices does not make them compatible.

The same distinction applies to generators. For each infinite
$L\subseteq\N_0$, the total set function $g_L(S)=\min(L\setminus S)$
always outputs a fresh point of $L$. Thus $\forall L\,\exists g_L$ holds
for every class; generation requires $\exists g\,\forall L$ instead.
\Cref{thm:normalization} preserves this order: the same $g$ is constructed
from $G$, before choosing the target-specific locking witnesses.
The functions $g_L$ need not be computable; Appendix~\ref{app:computability}
identifies the exact obstruction for a singleton.

\paragraph{Positive inputs versus exhaustive texts.}
Let $E=\{2,4,6,\ldots\}$, $O=\{1,3,5,\ldots\}$, and set
$L_E=\{0\}\cup E$, $L_O=\{0\}\cup O$.
These targets are generatable from exhaustive texts: after the first
nonzero observation, output the least unseen point of its parity.
Before then output $1$, and define the rule arbitrarily on inconsistent
histories containing both parities. An exhaustive text eventually shows
a nonzero point; equivalently, $\{2\}$ and $\{1\}$ are valid witnesses.

If inputs need only belong to the target, the constant stream
$0,0,\ldots$ is compatible with both languages. A common generator
would eventually have to output fresh points in both targets, yet
\[
 (L_E\setminus\{0\})\cap(L_O\setminus\{0\})=E\cap O=\varnothing.
\]
Exhaustiveness ensures that every finite positive witness eventually
appears. Positivity alone does not.

\Needspace{9\baselineskip}
\section{Computable generation versus separation width}
\label{app:computability}

Width records compatibility and cardinality of witnesses. Computability
asks whether fresh valid points can be obtained by a finite program.
The distinction already appears for a single target, whose width is
always zero.

Here the universe is $\N_0$ with its standard effective coding. Finite
words and finite sets have effective encodings, and a generator must be
total computable on \emph{all} finite histories. It receives no
membership oracle, target index, or correctness feedback. Success is
required on every exhaustive text, including noncomputable texts.
A set is \emph{computably enumerable} (c.e.) if an algorithm lists its
members; a set is \emph{decidable} if its membership predicate is total
computable.

\Needspace{19\baselineskip}
\begin{samepage}
\subsection{An exact criterion for one target}

\begin{proposition}[Computable generation for a singleton]
\label[proposition]{prop:computable-singleton}
For an infinite $L\subseteq\N_0$, the following are equivalent:
\begin{enumerate}[label=\textup{(\roman*)}]
 \item A total computable generator generates $L$ in the limit.
 \item There are a total computable
       $g:\fin{\N_0}\to\N_0$ and a finite $T\subseteq L$ such that
       \[
        T\subseteq S\subseteq L,\quad S\text{ finite}
        \quad\Longrightarrow\quad g(S)\in L\setminus S.
       \]
 \item $L$ contains an infinite c.e.\ subset.
 \item $L$ contains an infinite decidable subset.
\end{enumerate}
Moreover, whenever these conditions hold, some total computable
set-driven generator is correct and fresh on every finite positive
sample, with no initial errors.
\end{proposition}
\end{samepage}
\begin{proof}
\emph{(i)$\Rightarrow$(ii).}
Apply \cref{thm:normalization,cor:computable} with the standard
effective coding of $\N_0$. They supply one total computable set-driven
generator and a finite locking witness for this target.

\emph{(ii)$\Rightarrow$(iii).}
Fix one such finite $T$ and define recursively
\begin{equation}\label{eq:bootstrap}
 S_0=T,\qquad z_n=g(S_n),\qquad S_{n+1}=S_n\cup\{z_n\}.
\end{equation}
Inductively, $T\subseteq S_n\subseteq L$, so
$z_n\in L\setminus S_n$. The points $z_0,z_1,\ldots$ are therefore
distinct members of $L$. The recursion is computable: $T$ is a fixed
finite constant, and $g$ is total computable. Its outputs enumerate
an infinite c.e.\ subset of $L$.

\emph{(iii)$\Rightarrow$(iv).}
Given an enumeration of an infinite c.e.\ subset $W\subseteq L$,
wait for its first output $b_0$. Having found $b_n$, continue the
enumeration until a value $b_{n+1}>b_n$ appears. An infinite subset
of $\N_0$ is unbounded, so every search terminates. The resulting
sequence is computable and strictly increasing. Its range $B$ is
decidable: to decide whether $x\in B$, compute the sequence until
the first term at least $x$, then test equality.

\emph{(iv)$\Rightarrow$(i).}
For an infinite decidable $B\subseteq L$, set
\[
 g_B(S)=\min(B\setminus S),\qquad G_B(\sigma)=g_B(\cont(\sigma)).
\]
Membership in $B$ is decidable, and a finite $S$ cannot exhaust $B$,
so this defines total computable functions. Every output belongs to
$B\setminus S\subseteq L\setminus S$, proving both (i) and the final
assertion.
\end{proof}

The finite seed in \eqref{eq:bootstrap} is legitimate precisely because
the target is fixed. The proposition asserts the existence of a program;
it does not provide a procedure that recovers $T$ from an arbitrary
description of $L$ or $G$. Every particular finite set can be included
in a program as a constant. This use of finite advice is an existence
argument, separate from the established effective transformation
$G\mapsto g$.

Nor does the recursion pretend that its self-generated points exhaust
$L$. The normalization theorem supplies correctness on \emph{every}
finite extension of $T$ inside $L$. That stronger locking property is
what allows the outputs to be fed back as inputs.

\subsection{Width zero can coexist with a computational obstruction}

An infinite set is called \emph{immune} when it contains no infinite
c.e.\ subset; see, for example, \citet[Definition~1.1]{case2025immune}.
Thus \cref{prop:computable-singleton} says that an infinite target
admits computable generation exactly when it is not immune.
The obstruction can occur even when the target has density one.

\Needspace{10\baselineskip}
\begin{proposition}[A dense singleton with no computable generator]
\label[proposition]{prop:immune}
There exists a co-c.e.\ set $I\subseteq\N_0$ of density one such that
\[
 \spwidth(\{I\})=0
 \qquad\text{and}\qquad
 \{I\}\text{ has no total computable generator}.
\]
Here co-c.e.\ means that the complement is c.e.
\end{proposition}
\begin{proof}
List all c.e.\ sets as $W_0,W_1,\ldots$, allowing repetitions.
Dovetail their enumerations. For each $e$, if a point
$x>2^{e+1}$ ever appears in $W_e$, place the first such point into
a set $Q$ and take no further action for that index $e$.
This enumerates $Q$, with at most one contribution from each index.
For $n\ge2$,
\[
 |Q\cap[0,n)|\le |\{e\ge0:2^{e+1}<n\}|\le\log_2 n.
\]
Hence $I=\N_0\setminus Q$ is co-c.e., infinite, and has density one.

Every infinite $W_e$ eventually enumerates a point above $2^{e+1}$,
and the construction places such a point in $Q$.
Thus no infinite c.e.\ set is contained in $I$: the set $I$ is immune.
Apply \cref{prop:computable-singleton}. Its singleton width is zero
because the empty witness is valid for an infinite singleton target.
\end{proof}

Restricting success to computable texts would miss this example:
$I$ has no computable exhaustive text, since the range of such a text
would be an infinite c.e.\ subset of $I$.

The obstruction is the absence of an effectively accessible infinite
subset, rather than undecidability of the whole language. For example,
if $A\subseteq\N_0$ is any undecidable set, then
\[
 L_A=\{2n:n\in\N_0\}\cup\{2n+1:n\in A\}
\]
is undecidable, yet the least unseen even number is always a computable
valid output. Both $\{L_A\}$ and $\{I\}$ have width zero; only the former
has a computable generator.

\subsection{Individual algorithms need not combine}

The singleton criterion does not characterize computable generation
of a family by applying it to each target separately.

\begin{corollary}[A countable computational obstruction]
\label[corollary]{cor:computable-family}
Let $\mathcal R_1$ be the family of all decidable density-one subsets
of $\N_0$. Then $\spwidth(\mathcal R_1)=1$ and every member has a
computable generator, but no single total computable generator works
for the whole family.
\end{corollary}
\begin{proof}
There are only countably many decision programs, so \cref{thm:singleton}
gives width at most one. All cofinite sets belong to
$\mathcal R_1$, making its full intersection empty; the width is
therefore not zero. Every individual target is infinite and decidable,
so the minimum-selector construction applies to it.

Now suppose $G$ were a total computable generator for $\mathcal R_1$.
Carry out the construction in \cref{prop:density-one}. Every finite stage
is computable from $G$. Moreover, the resulting target $L$ is
decidable: given $x$, choose $r$ with $2^r\ge x$, simulate through
stage $r$, and test whether $x\in Q_r$. A point newly forbidden at
any later stage exceeds $2^r$, so
\[
 x\in L\quad\Longleftrightarrow\quad x\notin Q_r.
\]
Thus $L\in\mathcal R_1$, but the constructed text makes $G$ fail
infinitely often on $L$, a contradiction.
\end{proof}

This gives an effective counterpart to the quantifier distinction in
Appendix~\ref{app:quantifiers}. Individual decision procedures exist, while the
learner is not supplied with the unknown target's procedure. A countable
set of decidable languages need not come equipped with a uniform,
computable membership table covering its members.

\subsection{The role of oracle access}

The preceding results concern the stated no-oracle model.
In the membership-query model discussed by
\citet[Section~2.1, Remark~1]{charikar2024facets}, a learner may ask
whether a point belongs to a specified indexed language in the
collection. This does not identify which indexed language is the
unknown target. For a singleton collection, however, querying its sole
member is exactly access to target membership. Searching for the first
point of $L\setminus S$ then succeeds for every infinite $L$,
including immune targets. The obstruction in \cref{prop:immune}
is therefore compatible with results proved using membership queries.

The criterion in \cref{prop:computable-singleton} also has a precise
relative version: for any fixed oracle $A$, an infinite target has a
total $A$-computable generator if and only if it contains an infinite
$A$-c.e.\ subset. The proof relativizes each search and the normalization:
every candidate test uses finite data and calls to the given generator,
and the finite seed is still a constant. This keeps the oracle
fixed throughout, rather than granting access to the unknown target
silently.

\section{Formalization: design, correspondence, and lessons}
\label{app:formal}
The formal development makes explicit three interfaces in the argument:
from ordered histories to finite samples, from a mathematical search to
executable code, and from arbitrary witness assignments to finite counting
constraints. These interfaces help explain why the proofs work and which
parts can be reused. We describe the design choices, the clarifications
made when comparing the implementation with the manuscript, and the scope
of the resulting verification.
The additional density-one counterexample in
Appendix~\ref{app:quantifiers} and the computability statements in
Appendix~\ref{app:computability} are not included in this Lean development.

The public
\href{https://github.com/xiaoyulics/language-generation-characterization}{\texttt{language-generation-characterization}}
repository uses Lean~4 \citep{demoura2021lean4} and Mathlib
\citep{mathlib2020library}. It extends
\href{https://github.com/pengzhang91/generation-in-the-limit-lib}{\texttt{generation-in-the-limit-lib}},
developed and maintained by Shuangping Li and Peng Zhang. We retain their
four core definition files unchanged from upstream commit
\href{https://github.com/pengzhang91/generation-in-the-limit-lib/tree/de0d70c7e4645bface1d19bded9c8a5ade080fa8}{\texttt{de0d70c7}},
together with the Apache-2.0 license. The \texttt{FiniteWitness} modules are
our extensions; file-level provenance and credits are recorded in the
repository.

\subsection{Encoding the model without changing the question}
\label{app:formal-model}
A language is a \texttt{Set} of points, while a language class is a set of
such languages. This choice matters: indexing the class by natural numbers
would impose countability on the class being characterized. The main
theorem assumes a countably infinite point universe and that each target is
infinite; it places no countability assumption on the class and no
decidability assumption on target membership.

We use the upstream sequence-input definition directly. A generator takes
a length $t$ and an ordered history of that length. A positive text is a
stream whose range equals the target, so repetitions and arbitrary delays
are allowed. Correctness requires membership in the target and novelty
relative to the observed input; it does not require novelty relative to
earlier generated outputs. Finite words in the search are represented by
lists, and observed samples by \texttt{Finset}s. Lemmas connecting prefixes,
lists, and their underlying finite sets justify the changes of
representation. The existence of a set-driven generator is thus a proved
consequence of normalization; the original generator is allowed to depend
on order and repetitions.

The witness assignment is a function on all sets of points, with positivity
required only on members of the class. Values outside the class are
irrelevant. The definition \texttt{HasFiniteWitnesses} fixes this one
assignment before quantifying over all finite samples. It also retains the
nonemptiness guard on the active family. These details preserve both the
uniformity of the assignment and the empty-family cases of the paper.

\subsection{Normalization through a locking interface}
\label{app:formal-locks}
The main interface is \texttt{Locks g L}, which means
\[
 \exists T\in\fin{L}\quad
 \forall S\in\fin{L},\qquad
 T\subseteq S\ \Longrightarrow\ g(S)\in L\setminus S.
\]
It specifies correctness on every finite extension of a witness. This is
the property needed to pass from a generator to compatible witnesses.
For example, suppose the full intersection $C$ of a nonempty active family
at $S$ were finite. Then $S\subseteq C$, and applying every active target's
locking property to the finite sample $C$ would force $g(C)$ to lie both
inside and outside $C$. The lemma \texttt{locks\_imply\_finiteWitnesses}
encapsulates this argument. Conversely, an infinite active intersection
permits a fresh choice, yielding \texttt{finiteWitnesses\_imply\_locks}.
The normalization proof can therefore be developed independently of the
later width theory.

The quantifier order is essential. In the \texttt{Simplified} namespace,
\texttt{universal\_normalization} has the form
\[
 \forall G\ \exists g\ \forall L\in\infinite{\X},\qquad
 \bigl(G\text{ succeeds on every text for }L\bigr)
 \Longrightarrow\ \texttt{Locks}\ g\ L.
\]
Its proof constructs $g$ before introducing $L$. A separate construction
for each target would not supply the common generator required by the
characterization. The target-dependent objects enter only in the proof
that this already fixed $g$ locks.

We isolate the checkpoint mechanism in a structure \texttt{Checkpoints}.
The first-$k$ points in the chosen order provide an instance, but the
normalization proof uses only the properties in
\cref{tab:formal-checkpoints}. This identifies exactly what another
checkpoint construction would have to establish.

\begin{table}[H]
\centering
\small
\renewcommand{\arraystretch}{1.16}
\begin{tabular}{@{}p{.44\linewidth}p{.52\linewidth}@{}}
\toprule
Checkpoint property & Role in the proof \\
\midrule
$E_k(A)\subseteq A$ for every $A\subseteq\X$
& Checkpoints introduce only legal input points. \\
$|E_k(L)|=k$ for infinite $L$
& Seeing $E_{m+1}(L)$ guarantees at least $m+1$ search rounds. \\
Every $x\in L$ occurs in some $E_k(L)$
& An infinite chain of true errors would form an exhaustive text. \\
If $E_k(L)\subseteq S\subseteq L$, then
$E_j(S)=E_j(L)$ for $j\le k$
& The sample search and target-based error chain use the same checkpoints. \\
\bottomrule
\end{tabular}
\caption{The checkpoint interface used by the normalization proof.
Here $L$ is infinite and $S$ is finite. The first-$k$ construction is the
one in \cref{sec:normalization}.}
\label{tab:formal-checkpoints}
\end{table}

The proof then compares two recursions. The proof-only \texttt{trueRun}
uses membership in $L$ to extend a chain of genuine errors, while
\texttt{sampleRun} uses only the finite sample and the repaired generator.
The lemma \texttt{trueRun\_stops} rules out an infinite error chain.
The lemma \texttt{sampleRun\_matches\_true} shows that a sufficiently
informative sample makes the two histories agree through the final true
error. Its induction exposes the two reasons a smaller-code candidate
cannot interfere: a valid output has already been confirmed in the
sample, and an invalid output would contradict the minimality of the
true error.

One useful boundary case is $m=0$, when the true error chain stops at the
empty history. Terminal safety concerns strict extensions; it need not
make the output at the empty history correct. Including $E_{m+1}(L)$ in
the witness resolves this case as well as the general one: every sample
containing the witness has size at least $m+1$, so the normalized run
performs a strict extension beyond the terminal history. The cardinality
field of \texttt{Checkpoints} is thus a substantive part of the stopping
argument.

\subsection{From least-code search to executable code}
\label{app:formal-executable}
The mathematical search chooses the least code satisfying a predicate on
all finite words. We first formalize that specification, then give a
finite implementation on the natural-number universe and prove equality
with it. The key is to produce a candidate before searching for the best
candidate.

Let $F$ be the fresh repair, let $S$ be a nonempty sample, and suppose the
current history $p$ contains only points of $S$. Appending the sorted
sample gives
\[
 q_0=p\mathbin{{}^\frown}\operatorname{sort}(S),
 \qquad b=\code(q_0).
\]
This is a strict extension whose content is exactly $S$, so it contains
every required sample checkpoint and $F(q_0)\notin S$. Hence the globally
least candidate has code at most $b$. The executable definition decodes
the finitely many codes through $b$, filters by the candidate predicate,
and selects the least remaining code. Every test uses finite data and a
call to $F$; the target and its witness are absent from the computation.

The refinement is checked in stages. The lemma
\texttt{pick\_candidate\_eq} proves equality of one search step with the
least-code specification. Induction gives \texttt{executableRun\_eq},
and \texttt{executableNormalized\_eq} includes the empty sample, on which
both definitions perform zero rounds. The theorem
\texttt{executable\_universal\_normalization} then supplies the locking
conclusion. These equalities concern fixed, matching choices of encoding,
checkpoints, and fresh repair; they do not assert that all possible
normalizations produce identical outputs.

This supplies executable Lean definitions and checked correctness for the
finite-search argument in \cref{cor:computable}. The development does not
include a separate Mathlib \texttt{Computable}/\texttt{Partrec} theorem or
a verified compiler to Turing machines. It also gives no efficiency or
convergence bound. In particular, the older implementation retained in the
repository queries only words of length at most $2|S|$, whereas the present
implementation uses the adaptive code bound $b$ above. Their common
locking specification does not transfer that earlier query-length bound
to the revised algorithm.

\subsection{Preserving the scope of the width hierarchy}
\label{app:formal-width}
For the width values, the implementation uses
\texttt{WithTop (WithTop Nat)}. The inner added top represents $\omega$;
the outer one represents $\omega+1$. This supplies the ordered scale
$0,1,2,\ldots,\omega,\omega+1$ needed here without introducing general
ordinal arithmetic. The width is defined by the three cases in the paper.
The theorem \texttt{width\_isLeast\_cost} separately proves that this
definition agrees with the attained minimum of assignment costs,
including set-valued assignments and the empty class. Thus the convenient
representation is connected to the optimization meaning of width by a
theorem.

The finite lower bound requires particular care about the objects over
which it quantifies. Witnesses may depend on the entire target and may
contain arbitrary tail points. Restricting them to a fixed anchor pattern
would prove a weaker statement. The formal argument first uses capture to
derive the anchor-cover constraint for an arbitrary bounded assignment.
Only then does \texttt{Anchored.lower\_incidence} reduce the problem to
counting incidences in finite rows and columns. The final arithmetic is
short; the preceding reduction is what makes it apply to every admissible
assignment.

The direct capture proof also benefits from an explicit invariant. A
\texttt{DiagonalState} records a finite set $I$ of retained indices, a
finite set $B$ of forbidden points, and an infinite pool $P$ of future
indices, with
\[
 U_i\cap B=\varnothing\quad(i\in I),
 \qquad
 U_n\cap B=\varnothing\quad(n\in P).
\]
At the next stage, choose $d+1$ points of the next infinite core outside
the previously retained sets. Since every $U_n$ has at most $d$ points,
one candidate is omitted by infinitely many indices in $P$. Adding that
point to $B$ and restricting $P$ preserves safety both for past and for
future choices. A retained index larger than the stage makes the final
set of retained indices unbounded, hence infinite. This is slightly less
restrictive than choosing a strictly increasing sequence, and is enough
for capture. The theorem \texttt{bounded\_capture\_indexed} keeps the
point universe arbitrary: only the family of cores to be avoided must be
countable.

\subsection{What statement correspondence clarified}
\label{app:formal-lessons}
An instructive correction concerned the confirming witness in
\eqref{eq:normalization-witness}. The code takes one maximum of the
relevant history codes and collects all valid outputs below a common
finite bound. This gives a finite superset of the confirming outputs
listed in the displayed formula. The theorem
\texttt{exists\_confirming\_witness} proves the containments needed by
the proof, not literal equality with that formula. Comparing the
manuscript with the source led us to replace a comment describing an
``exact witness union'' with ``a sufficient finite confirming witness.''
The proof term did not change. This distinction is harmless for existence,
because enlarging a locking witness within the target preserves locking;
it would matter in an argument claiming a sharp witness-size bound.

The revised normalization raised a different correspondence issue. The
earlier proof used ambient-code checkpoints and a word-length cutoff.
Replacing them with first-$k$ checkpoints and exactly $|S|$ rounds
changes the algorithm. We therefore added the \texttt{Simplified} modules,
proved the first-$k$ interface, and checked the revised search separately.
The earlier theorem remains useful evidence for its own construction;
the new construction is supported by its own correctness and refinement
theorems. The direct diagonal capture proof was likewise checked
separately. Earlier hierarchy proof terms still invoke the older capture
proof by induction on the size bound; we do not identify those proof
terms with the new diagonal argument.

Other differences are recorded explicitly in the repository's
\href{https://github.com/xiaoyulics/language-generation-characterization/blob/0b3178f06d2285e8f2a7fa363c0f1732adea111e/docs/STATEMENT_MAP.md}{statement map}.
The finite upper bound uses a balanced two-block assignment in Lean and
a cyclic allocation in the text. The all-infinite endpoint uses capture
and the characterization in Lean, and a direct generator diagonal in the
text. The conclusion of \cref{cor:countable} follows from the checked,
stronger singleton theorem; its elementary index-based proof is not a
separate formal implementation. These are distinctions between verifying
a conclusion and verifying a particular proof or algorithm. Recording
them makes the relationship between paper and code inspectable.

\subsection{Checking and reusing the development}
\label{app:formal-reuse}
The mathematical release is pinned at commit
\href{https://github.com/xiaoyulics/language-generation-characterization/tree/0b3178f06d2285e8f2a7fa363c0f1732adea111e}{\texttt{0b3178f0}},
using Lean~4.24.0 and Mathlib commit
\href{https://github.com/leanprover-community/mathlib4/tree/f897ebcf72cd16f89ab4577d0c826cd14afaafc7}{\texttt{f897ebcf}}.
From that checkout, the build commands are
\begin{quote}
\small
\begin{verbatim}
cd GenLimitLean
lake exe cache get
lake build
\end{verbatim}
\end{quote}
The recorded successful build covers 44 local Lean source files; its
input manifest binds those files and three configuration files to their
hashes. The audit entry points print 41 theorem axiom closures, whose
union contains only \texttt{propext}, \texttt{Classical.choice}, and
\texttt{Quot.sound}. Classical choice is used in the information-theoretic
development, including choices from active intersections. Its presence
does not supply an algorithm for recovering witnesses from an arbitrary
class. The
\href{https://github.com/xiaoyulics/language-generation-characterization/blob/0b3178f06d2285e8f2a7fa363c0f1732adea111e/docs/CHECKING.md}{checking guide}
links the build record, provenance, and audit declarations.

The library also covers positive-separation equivalence, singleton
witnesses, the complete width hierarchy and its restriction law, chain
divergence, the dimension and local-profile obstructions, padding
collapse, EUC consequences, and the sorting counterexample. For reuse,
the principal entry points are
\begin{quote}
\small
\begin{verbatim}
import GenLimit.FiniteWitness.Simplified
import GenLimit.FiniteWitness.Width
\end{verbatim}
\end{quote}
A new normalization can target \texttt{Locks} and reuse the passage to
finite witnesses. A new checkpoint scheme can instantiate
\texttt{Checkpoints} and reuse the generic normalization proof. For a new
language class, proving \texttt{HasBoundedWitnesses H d} yields a width
upper bound through \texttt{width\_le\_finite\_iff}; proving optimality
requires excluding every assignment of smaller size, including witnesses
that depend on unrestricted parts of the target.

In this AI-assisted development, checker acceptance, statement
correspondence, and mathematical interpretation have distinct roles.
The build checks the encoded propositions under the declared axioms;
the statement map connects them to the manuscript; and the arguments
explain what the resulting notions measure. Keeping these connections
explicit lets readers check and extend the mathematics without treating
a successful build as evidence for claims about novelty, efficiency, or
an implementation that was never checked.

\end{document}